\documentclass[11pt]{article}
\usepackage{graphicx}
\usepackage{latexsym}
\usepackage{amssymb}
\usepackage{amsmath}
\usepackage{amsthm}
\usepackage{forloop}
\usepackage[paper=letterpaper,margin=1in]{geometry}
\usepackage[ruled,vlined,linesnumbered]{algorithm2e}
\usepackage{times}
\usepackage{paralist}
\usepackage{nicefrac}
\usepackage{rotating}
\usepackage{tikz}
\usepackage{array,multirow}
\usetikzlibrary{positioning}
\usepackage{threeparttable}
\usepackage{hyperref}
\usepackage{etoolbox}

\newbool{shortVersion}

\newtheorem{theorem}{Theorem}[section]
\newtheorem{lemma}[theorem]{Lemma}

\newtheorem{corollary}[theorem]{Corollary}
\newtheorem{definition}[theorem]{Definition}
\newtheorem{proposition}[theorem]{Proposition}
\newtheorem{property}[theorem]{Property}
\newtheorem{observation}[theorem]{Observation}

\makeatletter
\newtheorem*{rep@theorem}{\rep@title}
\newcommand{\newreptheorem}[2]{%
\newenvironment{rep#1}[1]{%
 \def\rep@title{#2 \ref{##1}}%
 \begin{rep@theorem}}%
 {\end{rep@theorem}}}
\makeatother

\newreptheorem{theorem}{Theorem}
\newreptheorem{reduction}{Reduction}
\newreptheorem{proposition}{Proposition}
\newreptheorem{lemma}{Lemma}

\DeclareMathOperator{\rank}{rank}
\DeclareMathOperator{\OPT}{OPT}

\newcommand\E{\mathbb{E}}

\newcommand{\defcal}[1]{\expandafter\newcommand\csname c#1\endcsname{{\mathcal{#1}}}}
\newcommand{\defbb}[1]{\expandafter\newcommand\csname b#1\endcsname{{\mathbb{#1}}}}
\newcounter{calBbCounter}
\forLoop{1}{26}{calBbCounter}{
    \edef\letter{\Alph{calBbCounter}}
    \expandafter\defcal\letter
		\expandafter\defbb\letter
}

\newcommand{\ie}{{\it i.e.}}
\newcommand{\eg}{{\it e.g.}}

\newcommand{\Greedy}{{\textsf{Greedy}}}
\newcommand{\Linear}{{\textsf{Linear}}}
\newcommand{\LinearPrime}{{$\textsf{Linear}\,'$}}
\newcommand{\SP}{{\textsf{MSP}}}
\newcommand{\SSP}{{\textsf{SMSP}}}
\newcommand{\MSSP}{{\textsf{MSMSP}}}
\newcommand{\SFM}{{\textsf{SFM}}}

\title{The Submodular Secretary Problem Goes Linear}

\author{Moran Feldman%
\thanks{School of Computer and Communication Sciences, EPFL. 
Email:
\href{mailto:moran.feldman@epfl.ch}{moran.feldman@epfl.ch}.
Supported by ERC Starting Grant 335288-OptApprox.}
\and
Rico Zenklusen%
\thanks{Department of Mathematics, ETH Zurich,
and Department of Applied Mathematics
and Statistics, Johns Hopkins University.
Email:
\href{mailto:ricoz@math.ethz.ch}{ricoz@math.ethz.ch}.
}}

\begin{document}

\maketitle

\begin{abstract}

During the last decade, the matroid secretary problem (\SP) became
one of the most prominent classes of online selection problems.
The interest in \SP\ is twofold: on the one hand, there are
many interesting applications of \SP; and on the other hand,
there is strong hope that \SP\ admits
$O(1)$-competitive algorithms,
which is the claim of the well-known matroid secretary
conjecture.
Partially linked to its numerous applications in mechanism design,
substantial interest arose also in the study of nonlinear versions
of \SP, with a focus on the submodular
matroid secretary problem (\SSP). 
The fact that submodularity captures the property of diminishing
returns, a very natural property for valuation functions,
is a key reason for the interest in \SSP.
So far, $O(1)$-competitive algorithms have been obtained
for \SSP\ over some basic matroid classes.
This created some hope that, analogously to the matroid secretary
conjecture, one may even obtain $O(1)$-competitive algorithms
for \SSP\ over any matroid.
However, up to now, most questions related to \SSP\ remained open,
including whether \SSP\ may be substantially more difficult
than \SP; and more generally, to what
extend \SP\ and \SSP\ are related.

Our goal is to address these points by presenting general
black-box reductions from \SSP\ to \SP.
In particular, we show 
that any $O(1)$-competitive algorithm for \SP,
even restricted to a particular matroid class,
can be transformed in a black-box way
to an $O(1)$-competitive algorithm for \SSP\ over
the same matroid class.
This implies that
the matroid secretary conjecture
is equivalent to the same conjecture for \SSP.
Hence, in this sense \SSP\ is not harder than \SP.
Also, to find $O(1)$-competitive algorithms for
\SSP\ over a particular matroid class, it suffices to
consider \SP\ over the same matroid class.
Using our reductions we obtain many first and improved
$O(1)$-competitive algorithms
for \SSP\ over various matroid classes by leveraging
known algorithms for \SP.
Moreover, our reductions imply an
$O(\log\log(\rank))$-competitive algorithm
for \SSP, thus, matching the currently best asymptotic
algorithm for \SP, and substantially improving on the previously
best $O(\log(\rank))$-competitive algorithm for \SSP.

\end{abstract}


\pagenumbering{alph}
\thispagestyle{empty}

\newpage

\pagenumbering{arabic}

\section{Introduction}

Secretary problems are a very natural class of online selection
problems with many interesting applications.
The origin of the secretary problem is hard to track and dates back
to at least the '60s~\cite{D63,
F89,%
G60,%
G60b,%
L61}.
In its original form, also called the \emph{classical secretary problem},
the task is to hire the best secretary out of a set $E$ of candidates
of known size $n=|E|$.
Secretaries get interviewed (or \emph{appear}) 
one by one in a random order. All secretaries that appeared so far
can be compared against each other
according to an underlying linear ordering.
Whenever a secretary got interviewed, one has to decide
immediately and irrevocably whether to hire (or \emph{select})
him. The task is to hire the best
secretary with as high a probability as possible.
Dynkin~\cite{D63} provided an
asympotically optimal algorithm for this
problem, which hires the best secretary with probability at
least $1/e$.
The classical secretary problem is naturally interpreted as a
stopping time problem and, not surprisingly, was mostly studied
by probabilists.

During the last decade, interest in generalized versions
of the classical secretary problem surged.
One reason for this is a variety of applications
in mechanism design 
(see~\cite{AKW14,%
BIKK08,%
BIK07,%
K05} and the references therein).
These generalizations allow hiring of more than
one secretary, subject to a given set of (down-closed)
constraints.
Each secretary reveals a non-negative
weight at appearance, and the task is to hire a maximum weight set
of secretaries.
The arguably most canonical generalization is the
problem of hiring $k$ out of $n$ secretaries instead of a
single one (see~\cite{K05}).
However, more general constraints are required for
many interesting applications.

A considerably more general setting, 
known as the \emph{matroid secretary problem}
({\SP} for short)
and introduced in~\cite{BIK07},
allows for selecting a subset of $E$
that is independent in a given matroid $\cM=(E,\mathcal{I})$.\footnote{
A matroid $\cM=(E,\mathcal{I})$ consists
of a finite set $E$ and a non-empty family
$\mathcal{I}\subseteq 2^E$ of subsets of $E$,
called \emph{independent sets} that satisfy:
(i) if $I\in \mathcal{I}$ and $J\subseteq I$ then $J\in \mathcal{I}$,
and (ii) if $I,J\in \mathcal{I}$ with $|I|>|J|$ then
there is an element $e\in I\setminus J$ such that
$J\cup \{e\} \in \mathcal{I}$.
For further basic matroidal concepts,
such as rank and span, we refer
to~\cite[Volume B]{S03}.
}
Similar to the classical secretary problem, the number
$n=|E|$ of candidates, or \emph{elements}, is known upfront,
elements appear in random order, and no assumption is made
on their weights.
Access to the matroid $\cM$ is provided by an independence
oracle that can be called on appeared elements, \ie, for
any subset $S\subseteq E$ of elements that appeared so far, one
can check whether $S\in \mathcal{I}$ or not.

\SP\ attracted considerable interest recently.
It is very appealing due to the fact that it captures 
a wide set of interesting selection problems in a single framework.
Furthermore, matroids are highly structured constraints,
which gives reasonable hope that strong online algorithms exist.
Indeed, there is a famous conjecture, which we simply
call the \emph{matroid secretary conjecture},
claiming the existence
of an $O(1)$-competitive algorithm
for the matroid secretary problem~\cite{BIK07}.
We recall that an algorithm is
$\alpha$-competitive if the expected weight collected by the algorithm
is at least $\frac{1}{\alpha} \cdot w(\OPT)$, where $w(\OPT)$ is the maximum
weight of any feasible set, \ie, the offline optimum.
Whereas this conjecture remains open, $O(1)$-competitive
algorithms have been obtained for various special cases of
matroids (see~\cite{%
BIK07,%
DP12,%
DK14,%
IW11,%
JSZ13,%
KP09,%
MTW13,%
S13b%
}).
The currently strongest asymptotic competitive ratio
obtained for general \SP---without any restriction
on the underlying matroid---is
$O(\log\log(\rank))$~\cite{lachish_2014_competitive,%
feldman_2015_simple},
where ``$\rank$'' is the rank of $\cM$, \ie, the cardinality of
a maximum cardinality independent set in $\cM$.

Recently, increased interest arose in nonlinear versions
of the secretary problem, with a focus on the
maximization of a non-negative submodular function\footnote{A non-negative submodular function $f$ on a ground set $E$ is a
function $f\colon2^E\rightarrow \bR^+$ giving
a non-negative weight to every subset of $E$ and satisfying
the following \emph{diminishing returns} property:
$f(A\cup\{e\}) - f(A) \geq f(B\cup \{e\}) - f(B)$ for
$A\subseteq B\subseteq E$ and $e\in E\setminus B$.}%
~\cite{barman_2012_secretary,bateni_2013_submodular,%
FNS11c,GRST10,MTW13}, leading
to the \emph{submodular secretary problem}.
Submodular functions have widespread use as valuation
functions because they 
reflect the property of diminishing returns,
\ie, the marginal value
of an element is the bigger the fewer elements have
been selected so far. This makes them natural candidates
for the matroid secretary setting.

Additionally, submodular weight functions capture
further generalizations of the secretary problem.
For example, Gilbert and Mosteller~\cite{gilbert_1966_recognizing}
and Freeman~\cite{freeman_1983_secretary} considered a
variation of the classical secretary problem where
one can select $k$ elements with the goal of maximizing
the value of the highest-valued element. This problem
can be phrased as a submodular secretary problem with
the submodular function $f:2^E\rightarrow \bR^+$
defined by $f(S) = \max\{w(e) \mid e\in S\}$ for
$S\subseteq E$, where
$w(e)$ is the weight revealed by element $e$.

The currently best asymptotic competitive ratio for the
submodular matroid secretary problem (\SSP) is
$O(\log(\rank))$~\cite{GRST10}. Furthermore,
$O(1)$-competitive algorithms have been obtained for
special classes of matroids, including
uniform matroids~\cite{bateni_2013_submodular,%
FNS11,GRST10}, partition matroids~\cite{FNS11,GRST10}, and
laminar and transversal matroids~\cite{MTW13} (both
only for monotone\footnote{A submodular function is monotone if $S \subseteq T \subseteq E$ implies $f(S) \leq f(T)$.} submodular functions).

In general, our understanding of
secretary problems is much more limited when
dealing with submodular weights instead of linear ones,
leading to many open questions.
In particular, is there hope to get an $O(1)$-competitive
algorithm for \SSP? Notice that this corresponds 
to the matroid secretary conjecture translated
to the submodular case.
Or may the submodular case be substantially harder
than the linear one?
Do monotone \SSP\ admit considerably
better competitive ratios than nonmonotone ones?
Can we leverage strong
algorithms for \SP\  to
obtain results for \SSP?

The goal if this paper is to address these questions 
and get a deeper understanding of \SSP\ 
and its relation to \SP, independently of
the structure of particular classes of
underlying matroids.

\subsection{Our Results} \label{sec:results}

Our main result below shows an intimate relation between the (linear)
matroid secretary problem and the submodular version.
More precisely, we show that one can use any
algorithm for \SP\ as a black box to obtain
an algorithm for \SSP\ with a slightly weaker
competitive ratio.
\begin{theorem} \label{thm:main_nonmonotone}
Given an arbitrary algorithm {\Linear} for {\SP} and
a value $\alpha \geq 1$, there exists an algorithm
for {\SSP} whose competitive ratio is at
most $24\alpha(3\alpha +1) = O(\alpha^2)$ for every
matroid $\cM$ on which {\Linear} is guaranteed to be
at least $\alpha$-competitive.
\end{theorem}
Theorem~\ref{thm:main_nonmonotone} has several interesting implications.
In particular, if there is an $O(1)$-competitive algorithm
for the linear case, then there is an $O(1)$-competitive
algorithm for the submodular case.
Hence, the matroid secretary conjecture is equivalent to
the same statement for the submodular version.
This provides strong hope that constant-competitive
algorithms exist for \SSP.

Furthermore, Theorem~\ref{thm:main_nonmonotone} 
implies many new results for \SSP, both for the general version
without any restriction on the matroid as well as for
many special classes of matroids, by leveraging algorithms
for \SP.
In particular, the known $O(\log\log(\rank))$-competitive
algorithms for \SP\ \cite{lachish_2014_competitive,%
feldman_2015_simple} imply
$O((\log\log(\rank))^2)$-competitive algorithms for
\SSP, which already considerably improves on the
previously best $O(\log(\rank))$-competitive
algorithm for \SSP~\cite{GRST10}.
We later strengthen this result to match the
asymptotically best algorithm for \SP.

The only matroid classes
for which $O(1)$-competitive
algorithms for \SSP\ have been explicitly given, without
assuming monotonicity
of the submodular weight function, are
uniform matroids~\cite{bateni_2013_submodular,GRST10}, unitary partition matroids~\cite{GRST10}
and matroids for which a reduction 
to unitary partition matroids
is known.
Such reductions are known for graphic
matroids~\cite{babaioff_2009_secretary,GRST10},
cographic
matroids~\cite{S13b}, and for
max-flow min-cut matroids~\cite{DK14}.
These reductions have originally been used
to obtain $O(1)$-competitive algorithms for {\SP} over
these matroids, but they
lead also to algorithms for {\SSP} over the
same matroid classes
when combined with an algorithm for {\SSP} over
unitary partition matroids.
For other classes of matroids,
Theorem~\ref{thm:main_nonmonotone} implies 
the first $O(1)$-competitive algorithm for \SSP,
by leveraging known $O(1)$-competitive algorithms for \SP,
such as the ones known for
transversal matroids~\cite{KP09} and
laminar matroids~\cite{IW11}.
Furthermore, we also improve the competitive ratios
for most matroid classes for which
$O(1)$-competitive algorithms have already been known.

The analysis of the algorithm that we use
to prove Theorem~\ref{thm:main_nonmonotone} can easily
be improved in many cases, if the algorithm 
{\Linear} obeys some natural properties.
Theorem~\ref{thm:union_nonmonotone} below gives a first
strengthening of Theorem~\ref{thm:main_nonmonotone}
and allows us to derive an $O(\log\log(\rank))$-algorithm
for \SSP, thus, matching the currently best algorithm for \SP\
up to a constant factor, and improving on the previously
best $O(\log (\rank))$ competitive algorithm for \SSP.

We highlight that Theorem~\ref{thm:main_nonmonotone}
can be obtained from Theorem~\ref{thm:union_nonmonotone}
by setting $k=1$.

\begin{theorem} \label{thm:union_nonmonotone}
Given an algorithm {\Linear} for {\SP} and value $\alpha \geq 1$,
there exists an algorithm for {\SSP} whose competitive
ratio is at most $24 k \alpha(3\alpha+1)=O(k\cdot \alpha^2)$
for every matroid $\cM$ on which the behavior
of {\Linear} can be characterized as follows.
\begin{compactitem}
	\item For every {\SP} instance over the matroid $\cM$, there exists a set of $k$ (correlated) random sets $\{P_i\}_{i = 1}^k$ such that each set $P_i$ is always independent in $\cM$ and $\bE[w(\bigcup_{i = 1}^k P_i)] \geq w(\OPT_w) / \alpha$, where $w$ is the weight function of the {\SP} instance and $\OPT_w$ is the maximum weight independent set given $w$.
	\item {\Linear} outputs a uniformly random set from $\{P_i\}_{i = 1}^k$.
\end{compactitem}
Observe that such an algorithm {\Linear} is
$(k\cdot \alpha)$-competitive for $\cM$.
\end{theorem}

\begin{corollary}\label{corr:loglogSMSP}
There exists an $O(\log \log(\rank))$-competitive algorithm
for {\SSP}.
\end{corollary}
\begin{proof}
Feldman et al.~\cite{feldman_2015_simple} describe
an $O(\log \log(\rank))$-competitive algorithm for {\SP}.
Their analysis of the algorithm shows that in fact it obeys
the requirements of Theorem~\ref{thm:union_nonmonotone} for
{\Linear} with $\alpha = O(1)$ and $k = O(\log \log (\rank))$.
Notice that the algorithm guaranteed by Theorem~\ref{thm:union_nonmonotone} does not depend on $k$. Hence, there is no need to know the $\rank$ ahead of time to use this algorithm.
\end{proof}

\medskip

In what follows we present further useful strengthenings of
Theorem~\ref{thm:main_nonmonotone} that are readily obtained
through our derivation of Theorem~\ref{thm:main_nonmonotone}.
These results allow for obtaining stronger competitive
ratios if {\Linear} fulfills some very typical properties,
or if the submodular valuation function is monotone.
When dealing with monotone submodular functions, we talk
about the \emph{monotone submodular secretary problem} (\MSSP).
We first state the improved reductions and later
present their implications in terms of competitive
ratios in the form of a table.

\begin{theorem} \label{thm:bound_probability_nonmonotone}
Given an algorithm {\Linear} for {\SP} and values $\alpha\geq 1$ and $q \in (0, 1]$, there exists an algorithm for {\SSP} whose competitive ratio is at least $24\alpha (3q\alpha +1) = O(q\cdot \alpha^2)$ for every matroid $\cM$ on which {\Linear} is guaranteed to:
\begin{compactitem}
	\item be at least $\alpha$-competitive.
	\item select every element with probability at most $q$.
\end{compactitem}
\end{theorem}

The results proved above can be somewhat improved for monotone objective functions.
\ifbool{shortVersion}{}{Appendix~\ref{app:monotone} explains
how the proofs of Theorems~\ref{thm:main_nonmonotone}, \ref{thm:union_nonmonotone} and~\ref{thm:bound_probability_nonmonotone} need to be changed in order to get the following theorem.
}

\begin{theorem} \label{thm:monotone}
When {\SSP} is replaced with {\MSSP} in Theorems~\ref{thm:main_nonmonotone}, \ref{thm:union_nonmonotone} and~\ref{thm:bound_probability_nonmonotone} their guarantees can be improved to $8\alpha(\alpha+1)$, $8 k \alpha(\alpha+1)$,
and $8\alpha (\alpha q +1)$, respectively.
\end{theorem}

For monotone functions there is one additional natural property of {\Linear} that can be used to get a stronger result. Intuitively this property is that {\Linear} selects every element of the optimal solution with probability at least $\alpha^{-1}$, and thus, is $\alpha$-competitive. Many algorithms have this property when items have disjoint weights, and thus, there is a single optimal solution. Such algorithms are often extended to general inputs by introducing a random tie breaking rule%
.
The following theorem is designed to deal with algorithms obtained this way%
\ifbool{shortVersion}{.}{; its proof can be found in Appendix~\ref{app:monotone}.}

\begin{theorem} \label{thm:opt_probability_monotone}
Given an arbitrary algorithm {\Linear} for {\SP} and a value $\alpha \geq 1$, there exists an algorithm for {\MSSP} whose competitive ratio is at most
$16 \alpha = O(\alpha)$ for every matroid $\cM$ on which {\Linear} has the following property: For every {\SP} instance over the matroid $\cM = (\cN, \cI)$, there exists a random set $S \subseteq E$ obeying:
\begin{compactitem}
	\item $S$ is always an optimal solution of the {\SP} instance.
	\item For every element $u \in E$, $\Pr[\text{$u$ is selected by Linear}] \geq \Pr[u \in S] / \alpha$.
\end{compactitem}
\end{theorem}

The following table summarizes the competitive ratios 
we obtain,
for {\SSP} and {\MSSP} over particular matroid classes, 
by leveraging the above-presented reductions.\footnote{For readability, the competitive ratios in the table have been rounded to the closest integer.}
A straightforward application of these reductions leads to improvements or even
first results for all matroid classes listed in the table.
However, using some additional observations we can sometimes
get further improvements.
\ifbool{shortVersion}{}{%
Further details on how the stated results are implied
by our reductions can be found in Appendix~\ref{app:specific_matroids}.
}%
Our improvement for unitary partition matroids also
implies improvements by the same factor for all matroid
classes for which a reduction to
unitary partition matroids is known.

\begin{center}
\begin{threeparttable}
\begin{tabular}{|c|>{\hspace{2.5em}\centering}m{3.5cm}>{\centering}m{1.2cm}|c|}
\hline
Matroid Type & \multicolumn{2}{c|}{Known Competitive Ratio} & Our Competitive Ratio \\ \hline
Unitary partition matroids & $1297$ \tnote{$(1)$}
& \centering \cite{GRST10}& $261$ \\
Transversal matroids 			 & $-$ &  & $2496$ \\
$k$-sparse linear matroids \tnote{$(2)$} & $-$ & & $24 k e (3k+1)$ \\
$k$-sparse linear matroids (\MSSP) & $-$ & & $8 k e (k+1)$ \\
Laminar matroids           & $-$ &  & $585$ \\
Laminar matroids (\MSSP)   & $211$ & \cite{MTW13} &
  $144$ \\
\hline
\end{tabular}
\begin{tablenotes}
	\item[$(1)$] Gupta et al.~\cite{GRST10} do not explicitly calculate the competitive ratio of their algorithm, however, its competitive ratio is no better than $48000/37 \approx 1297$.

\item[$(2)$] A $k$-sparse linear matroid is a linear matroid
with a matrix representation using at most $k$ non-zeros
per column.

\end{tablenotes}
\end{threeparttable}
\end{center}

We highlight that apart from the unitary partition matroid
case, all other results in the above table
for particular matroid classes 
assume prior knowledge of the matroid, in addition to the
size of its ground set.
This is due to the fact that the corresponding \SP\ algorithms
that we put into our framework to get results for \SSP\ make
this assumption.

\medskip

\noindent \textbf{Remark:} The proofs of all the above theorems
use {\Linear} in a black-box manner. Hence, these theorems
apply also to many models allowing the algorithm more information,
such as the model in which the algorithm has full knowledge about
the matroid from the beginning (but not about the objective function).
Also, we emphasize that all the above results are general reductions
from \SP\ to \SSP, not assuming any particular structure
about the underlying matroid. This is in stark contrast
to almost all results on \SSP\ so far.

\subsection{Further Related Work}

Progress has been made on the matroid secretary
conjecture for variants of {\SP} which modify the
assumptions on the order in which elements arrive
and the way weights are assigned to elements.
One simpler variant of \SP\ is obtained by assuming
\emph{random weight assignment}. 
Here, an adversary can only choose $n=|E|$
not necessarily distinct weights, and the weights
are assigned to the elements uniformly at random.
In this model, a $5.7187$-competitive algorithm
can be obtained for any matroid~\cite{S13b}.
Additionally, a $16(1-1/e)$-competitive procedure
can still be obtained
in the random weight assignment model even if
elements are assumed to arrive in
an adversarial order~\cite{GV11,S13b}.
Hence, this variant is, in a sense, the opposite of
the classical \SP, where weights are adversarial
and the arrival order is random.
Furthermore, a $4$-competitive algorithm can be obtained
in the so-called \emph{free order model}. Here,
the weight assignment is adversarial; however,
the algorithm can choose the order in which
elements
arrive~\cite{AKW14,JSZ13}.
Among the above-discussed variants, this is the only
variant with adversarial weight assignments for which
an $O(1)$-competitive algorithm is known.
For more information on recent advances on \SP\ and
its variants we refer to the
survey~\cite{dinitz_2013_recent}.

We also highlight that 
\SSP\ is an online version of submodular
function maximization (\SFM)
over a matroid constraint.
Interestingly,
even in the offline setting,
$O(1)$-approximations for {\SFM}
over a matroid constraint have only been discovered
very recently, starting with
a $(4+\epsilon)$-approximation
presented in~\cite{LMNS10}.
Considerable progress has been made 
in the meantime~\cite{GV11,FNS11}.
The currently strongest approximation algorithm
has an approximation ratio of about
$e \approx 2.718$~\cite{FNS11}.
We refer the interested reader
to~\cite{chekuri_2014_submodular,FNS11} for more
information on constrained {\SFM}.

\subsection{Organization of the paper}
We start by formally introducing our problem and some
basic notation and results in Section~\ref{sec:preliminaries}.
Section~\ref{sec:algorithm} presents our main algorithm that
we use to prove our results. Finally, Section~\ref{sec:analysis}
provides details on the analysis of our algorithm.

\ifbool{shortVersion}{
Due to space constraints, many proofs have been omitted
from this extended abstract. All omitted proofs can 
be found in the long version of this paper.
}{}

\section{Preliminaries} \label{sec:preliminaries}

In this section we formally define our problem and state some notation and known results that are used later in the paper.

\subsection{Problems and Standard Notation}

An instance of the Submodular Matroid Secretary Problem (\SSP) consists of a ground set $E$, a non-negative submodular objective $f\colon E \to \bR^+$ and a matroid constraint $\cM = (E, \cI)$. An algorithm for this problem faces the elements of $E$ in a uniformly random order, and must accept or reject each element immediately upon arrival. The algorithm has access to $n = |E|$ and two oracles:
\begin{compactitem}
	\item A value oracle that, given a subset $S \subseteq E$ of elements that \emph{already arrived}, returns $f(S)$.
	\item An independence oracle that, given a subset $S \subseteq E$ of elements that \emph{already arrived}, determines whether $S \in \cI$.
\end{compactitem}
The objective of the algorithm is to accept an independent set of elements maximizing $f$. 

The Matroid Secretary Problem (\SP) can be viewed as a restriction of \SSP\ to linear objective functions. More formally, an instance of {\SP} is an instance of {\SSP} in which the function $f(S)$ is defined by $f(S) = \sum_{u \in S} w(u)$ for some set of non-negative weights $\{w(u) \mid u \in E\}$. Similarly, the Monotone Submodular Matroid Secretary Problem (\MSSP) is a restriction of \SSP\ to non-negative monotone submodular objective functions.

The following notation comes handy in our proofs. Given a set $S$ and an element $u$, we use $S + u$ and $S - u$ to denote the sets $S \cup \{u\}$ and $S \setminus \{u\}$, respectively. Additionally, given a weight function $w \colon E \to \bR^+$, we use $w(S)$ as a shorthand for $\sum_{u \in S} w(u)$. Finally, given a set function $f \colon 2^E \to \bR$, we denote by $f(u \mid S)$ the marginal contribution of adding $u$ to $S$. More formally, $f(u \mid S) := f(S + u) - f(S)$.

\subsection{The Function \texorpdfstring{$f_w$}{f\_w}}

Given a set function $f\colon 2^E \to \bR$ and a weight vector $w \in \bR^E$, let $f_w\colon 2^E \to \bR$ be the function defined as:
\[
	f_w(S)
	=
	\min_{A \subseteq S} \{f(A) + w(S \setminus A)\}
  \qquad \forall S\subseteq E \enspace.
\]
This construction of $f_w$ out of $f$ and $w$ is well known
in the field of submodular function optimization, and is
sometimes called \emph{convolution}
(see, for example,~\cite{lovasz_1983_submodular}).
We state some basic properties of $f_w$.
Proofs of these results can, for example, be
found in~\cite[Chapter~10]{narayanan_1997_submodular}.
\begin{property}\label{prop:w_always_smaller}
For every set $S \subseteq E$, $f(S) \geq f_w(S)$, and for $S = \varnothing$ the inequality holds with equality.
\end{property}
\begin{property}\label{prop:properties_preserved}
If $f$ is a non-negative submodular function and $w$ is non-negative, then $f_w$ is also non-negative and submodular. Moreover, if $f$ is also monotone, then so is $f_w$.
\end{property}

\ifbool{shortVersion}{}{
\subsection{Known Lemmata}

We need the following known lemmata. The lemmata have been rephrased a bit to make the difference between them clearer.

\begin{lemma}[Lemma~2.2 of~\cite{FMV11}] \label{lem:sampling}
Let $g\colon 2^E \to \bR^+$ be a submodular function. Denote by $A(p)$ a random subset of $A$ where each element appears with probability $p$ (not necessarily independently). Then,
\[
	\bE[g(A(p))] \geq (1 - p) g(\varnothing) + p \cdot g(A)
	\enspace.
\]
\end{lemma}

\begin{lemma}[Lemma~2.2 of~\cite{BFNS14}] \label{lem:max_prob_bound}
Let $g\colon 2^E \to \bR^+$ be a non-negative submodular function. Denote by $A(p)$ a random subset of $A$ where each element appears with probability at most $p$ (not necessarily independently). Then,
\[
	\bE[g(A(p))] \geq (1 - p)g(\varnothing)
	\enspace.
\]
\end{lemma}
}

\section{Algorithm} \label{sec:algorithm}

The algorithm used to prove our results for {\SSP} is given as Algorithm~\ref{alg:online}. Observe that the algorithm has a single probability parameter $p \in (0, 1)$. Additionally, the algorithm uses an arbitrary procedure {\Linear} for {\SP} whose existence is assumed by Theorem~\ref{thm:main_nonmonotone} and our other results for {\SSP}. Finally, the algorithm also uses as a subroutine the standard greedy algorithm for maximizing a submodular function subject to matroid constraints (denoted by {\Greedy})%
\ifbool{shortVersion}{.}{, 
which can be found, for completeness, in Appendix~\ref{app:equivalence}.}%
\footnote{For non-monotone functions there are two versions of {\Greedy}: one that stops picking elements once they all have negative marginal contributions, and one that continues as long as possible. In this paper we use the first version.} While reading the description of Algorithm~\ref{alg:online}, the only important thing one has to know about {\Greedy} is that it creates its solution by starting with the empty set and adding elements to it one by one.

A key challenge in trying to leverage {\Greedy} in an algorithm for {\SSP}, is that {\Greedy} is not a constant-factor approximation algorithm for submodular function maximization over a matroid (or even in the unconstrained setting). In our analysis we show that Algorithm~\ref{alg:online} manages to circumvent this issue, and the set $M$, produced within Algorithm~\ref{alg:online} using {\Greedy}, does provide a constant approximation for the optimal solution.

\SetKwIF{With}{OtherwiseWith}{Otherwise}{with}{do}{otherwise with}{otherwise}{}
\begin{algorithm}[ht]
\caption{\textsf{Online}$(p)$} \label{alg:online}
\tcp{Learning Phase}
Choose $X$ from the binomial distribution $B(n, 1/2)$.\\
Observe (and reject) the first $X$ elements of the input. Let $L$ be the set of these elements.\\

\BlankLine

\tcp{Selection Phase}
Let $M$ be the output of {\Greedy} on the set $L$.\\
Let $N \gets \varnothing$.\\
\For{each arriving element $u \in E \setminus L$}
{
	Let $w(u) \gets 0$.\\
	\If{$u$ is accepted by {\Greedy} when applied to $M + u$}
	{
		\With{probability $p$}
		{
			Add $u$ to $N$.\\
			Let $M_u \subseteq M$ be the solution of {\Greedy} immediately before it adds $u$ to it.\\
			Update $w(u) \gets f(u \mid M_u)$.\\
		}
	}
	Pass $u$ to {\Linear} with weight $w(u)$.
}
\Return{$Q \cap N$, where $Q$ is the output of {\Linear}}.
\end{algorithm}

Observe that Algorithm~\ref{alg:online} can be implemented online because whenever an element $u$ is fed to {\Linear}, the algorithm already knows whether $u \in N$, and thus, can determine the membership of $u$ in $Q \cap N$ immediately after {\Linear} determines the membership of $u$ in $Q$.
Additionally, note that Algorithm~\ref{alg:online} applies {\Linear} to 
the restriction\footnote{The \emph{restriction} of a matroid $\cM=(E,\cI)$
to a subset $E'\subseteq E$ of its ground set is the
matroid $\cM' = (E', \cI \cap 2^{E'})$.}
of its input matroid $\cM$ to the set $E \setminus L$. We assume in the analysis of Algorithm~\ref{alg:online} that {\Linear} is $\alpha$-competitive for that restriction of $\cM$ whenever it is $\alpha$-competitive for $\cM$. Many algorithms for {\SP} obey this property without any modifications, but for some we need the following proposition which proves that this assumption can be justified in general. 

\begin{definition}
Partial-{\SP} is a variant of {\SP} where an instance consists also of a set $L \subseteq E$ \emph{known} to the algorithm. The elements of $E \setminus L$ arrive at a uniformly random order as usual. The elements of $L$ never arrive, and thus, cannot be added to the solution. However, the oracles can answer queries about them as if they arrived before the first element of $E \setminus L$.
\end{definition}

\begin{proposition} \label{prp:partial_problem}
Given an algorithm {\Linear} for {\SP}, there exists an algorithm for Partial-{\SP} whose competitive ratio for every matroid $\cM$ is as good as the competitive ratio of {\Linear} for this matroid.
\end{proposition}

\ifbool{shortVersion}{}{The proof of Proposition~\ref{prp:partial_problem} can be found in Appendix~\ref{app:missingProofs}.}
Notice that {\Linear} indeed faces an instance of Partial-{\SP} since the set $L$ is fully known before {\Linear} is invoked for the first time, and the weights assigned to elements depend solely on $L$.
The following simple observation is well known. A proof of it can be found, \eg, as Lemma~A.1. of~\cite{feldman_2015_simple}.

\begin{observation}
  \label{obs:LpropP}
The set $L$ constructed by Algorithm~\ref{alg:online} contains every element of $E$ with  probability $1/2$, independently.
\end{observation}

As is, it is difficult to prove some claims about Algorithm~\ref{alg:online}. For that purpose, we present Algorithm~\ref{alg:simulated}, which is an offline algorithm sharing the same joint distribution (as Algorithm~\ref{alg:online}) of the set $M$ and the output set (a justification of this claim can be found in
\ifbool{shortVersion}{the long version of the paper).}{%
Appendix~\ref{app:equivalence}).}
Clearly, the competitive ratio of Algorithm~\ref{alg:online} is equal to the approximation ratio of Algorithm~\ref{alg:simulated}.

\begin{algorithm}[ht]
\caption{\textsf{Simulated}$(p)$} \label{alg:simulated}
\DontPrintSemicolon
\tcp{Initialization}
Let $M, N, N_0 \gets \varnothing$.\\
Let $E' \gets E$.\\

\BlankLine

\tcp{Main Loop}
\While{$E' \neq \varnothing$}
{
	Let $u$ be the element of $E'$ maximizing $f(u \mid M)$, and remove $u$ from $E'$.\\
	Let $w(u) \gets f(u \mid M)$.\\
	\If{$M + u$ is independent in $\cM$ and $w(u) \geq 0$}
	{
		\lWith{probability $1/2$}
		{
			Add $u$ to $M$.
		}
		\lOtherwiseWith{probability $p$}
		{
			Add $u$ to $N$.
		}
		\Otherwise
		{
			Update $w(u) \gets 0$.\\
			Add $u$ to $N_0$.
		}
	}
	\Else
	{
		Let $w(u) \gets 0$.\\
		Add $u$ to $N_0$ with probability $1/2$.
	}
}
Run {\Linear} with $N \cup N_0$ as the input (in a uniformly random order) and the weights defined by $w$.\\
\Return{$Q \cap N$, where $Q$ is the output of {\Linear}}.
\end{algorithm}

Notice that Algorithm~\ref{alg:simulated} also makes {\Linear} face an instance of Partial-{\SP} (with $L = E \setminus (N \cup N_0)$). Hence, by Proposition~\ref{prp:partial_problem}, we can assume the competitive ratio of {\Linear} for $\cM$ extends to the instance it faces in Algorithm~\ref{alg:simulated}.

Algorithms~\ref{alg:online} and~\ref{alg:simulated} can be viewed as generalizations of algorithms presented by~\cite{MTW13} for {\MSSP} over laminar matroids.
In particular, the idea of defining a surrogate weight function $w$ based on a subset of the elements sampled at the beginning
of the algorithm was already used in~\cite{MTW13}.
Although our analysis of Algorithms~\ref{alg:online} and~\ref{alg:simulated} is quite different and much more general than the analysis of~\cite{MTW13}, it does borrow some ideas from~\cite{MTW13}.
The concept of an offline algorithm simulating a more difficult to analyze online algorithm has been previously used, even specifically for offline simulations of an online greedy algorithm~\cite{DP12,KP09,MTW13}.
A key novel contribution of our analysis, compared
to~\cite{MTW13}, is that we manage to relate
the expected $w$-weight of a maximum $w$-weight independent set 
in $N$ to $\E[w(M)]$ (see Lemma~\ref{lem:opt_n_bound}).
Furthermore, we overcome several technical hurdles by
first comparing values of constructed sets with respect
to the convoluted submodular function $f_w$ instead of 
the original function $f$.
Finally, our use (and analysis) of a modified greedy
algorithm allows us to deal with
non-monotone submodular functions.

\section{Analysis of Algorithms~\ref{alg:online} and~\ref{alg:simulated}} \label{sec:analysis}

Throughout this section
\ifbool{shortVersion}{}{
(except in Section~\ref{ssc:union_analysis})
}
we fix an arbitrary matroid $\cM = (E, \cI)$ for which {\Linear} is $\alpha$-competitive, and analyze the approximation ratio of Algorithm~\ref{alg:simulated} for this matroid. Since Algorithms~\ref{alg:online} and~\ref{alg:simulated} share their output distribution, the approximation ratio we prove for Algorithm~\ref{alg:simulated} implies an identical competitive ratio for Algorithm~\ref{alg:online}. The analysis of the approximation ratio consists of three main stages. In each stage we study one of the sets $M$, $N$ and $N \cap Q$. Specifically, we show bounds on the expected values assigned to each one of these sets by $w$ and $f_w$. Notice that once we have a bound on $\bE[f_w(N \cap Q)]$, we also get a bound on the expected value of the solution produced by Algorithm~\ref{alg:simulated} since $f(N \cap Q) \geq f_w(N \cap Q)$ by Property~\ref{prop:w_always_smaller}.\footnote{Observe that the weights $w$ chosen by Algorithm~\ref{alg:simulated} are always non-negative, thus, we can use all the properties of $f_w$ stated in Section~\ref{sec:preliminaries}.}

\ifbool{shortVersion}{}{
Following is some notation that is useful in many of the proofs below. For every element $u \in E$, let $E_u$ be the set of elements processed by Algorithm~\ref{alg:simulated} before $u$. Then, we define:
\[
	N_u = N \cap E_u \qquad \mbox{and} \qquad \enspace M_u = M \cap E_u \enspace.
\]
Similarly, we also define $E'_u = E_u + u$ and:
\[
	N'_u = N \cap E'_u \qquad \mbox{and} \qquad \enspace M'_u = M \cap E'_u\enspace.
\]
}

\subsection{First Stage}

We now begin with the first stage of the analysis, \ie, bounding $\bE[w(M)]$ and $\bE[f_w(M)]$. The following lemma shows that both values are in fact strongly related to $\bE[f(M)]$.

\begin{lemma} \label{lem:M_equalities}
$w(M) + f(\varnothing) = f(M) = f_w(M)$.
\end{lemma}
\ifbool{shortVersion}{}{
\begin{proof}
Observe that, by construction, $w(u) = f(u \mid M_u)$ for every $u \in M$. Hence, the equality $w(M) + f(\varnothing) = f(M)$ holds since:
\[
	w(M)
	=
	\sum_{u \in M} f(u \mid M_u)
	=
	f(M) - f(\varnothing)
	\enspace.
\]

Let us now prove the equality $f_w(M) = f(M)$. By Property~\ref{prop:w_always_smaller}, $f_w(M) \leq f(M)$. Thus, we only need to prove the reverse inequality. For every set $A \subseteq M$:
\begin{align*}
	f(A) + w(M \setminus A)
	={} &
	f(\varnothing) + \sum_{u \in A} f(u \mid A \cap M_u) + \sum_{u \in M \setminus A} f(u \mid M_u)\\
	\geq{} &
	f(\varnothing) + \sum_{u \in A} f(u \mid M_u) + \sum_{u \in M \setminus A} f(u \mid M_u)
	=
	f(M)
	\enspace,
\end{align*}
where the inequality follows from submodularity. Hence, by the definition of $f_w(M)$:
\[
	f_w(M)
	=
	\min_{A \subseteq M} \{f(A) + w(M \setminus A)\}
	\geq
	\min_{A \subseteq M} f(M)
	=
	f(M)
	\enspace.
	\qedhere
\]
\end{proof}
}

The bounds given by the above lemma are in terms of $f(M)$. To make these bounds useful, we need to bound also $\bE[f(M)]$. This is not trivial since {\Greedy} is not a constant-factor
approximation algorithm for submodular function maximization over
a matroid.
Different approaches are known to adapt or
extend {\Greedy} such that it provides
an $O(1)$-approximation in the offline setting.
However, these are not well-suited for the way we
simulate greedy online.
The next \ifbool{shortVersion}{lemma shows}{two lemmata show} a very simple way
to transform {\Greedy}
into an $O(1)$-approximation algorithm for submodular function
maximization; and most importantly, this adjustment
of {\Greedy} is trivial to simulate online.
We are not aware of any previously
known variation of {\Greedy}
that provides an $O(1)$-approximation for submodular function
maximization over a matroid constraint,
and that can easily be simulated in the online setting. 

\ifbool{shortVersion}{}{
The next lemma is similar to Lemma~3 of~\cite{GRST10}, but does not assume that $f$ is normalized (\ie, $f(\varnothing) = 0$).
The proof of the lemma is deferred to Appendix~\ref{app:missingProofs} since it goes along the same lines as the proof of~\cite{GRST10}.

\begin{lemma} \label{lem:greedy_determinisitic_property_nonnormalized}
For a matroid and a non-negative submodular function $f$, if $S$ is the independent set returned by {\Greedy}, then for any independent set $C$, $f(S) \geq f(C \cup S) / 2$.
\end{lemma}

The following lemma gives the promised variant of {\Greedy}.
}

\begin{lemma} \label{lem:greedy_property}
Let $S$ be a random set containing every element of $E$ with probability $1/2$, independently, then $\bE[f(\Greedy(S))] \geq f(\OPT) / 8$.
\end{lemma}
\ifbool{shortVersion}{}{
\begin{proof}
Let $T$ be the output of {\Greedy}. By Lemma~\ref{lem:greedy_determinisitic_property_nonnormalized}, the following inequality always holds:
\[
	f(T)
	\geq
	\frac{f(T \cup (\OPT \cap S))}{2}
	\enspace.
\]
Hence,
\begin{align*}
	\bE[f(T)]
	={} &
	\sum_{A \subseteq \OPT} \Pr[\OPT \cap S = A] \cdot \bE[f(T) \mid \OPT \cap S = A]\\
	\geq{} &
	\frac{1}{2} \cdot \sum_{A \subseteq \OPT} \Pr[\OPT \cap S = A] \cdot \bE[f(T \cup A) \mid \OPT \cap S = A]
	\enspace.
\end{align*}

Let $g_A(T) = f(T \cup A)$. It is easy to check that $g_A$ is a non-negative submodular function for every choice of set $A$. Additionally, $T \subseteq S$, and thus, contains every element of $E \setminus A$ with probability at most $1/2$ even conditioned on $\OPT \cap S = A$. Hence, by Lemma~\ref{lem:max_prob_bound}:
\[
	\bE[f(T \cup A) \mid \OPT \cap S = A]
	=
	\bE[g_A(T \setminus A) \mid \OPT \cap S = A]
	\geq
	\frac{g_A(\varnothing)}{2}
	=
	\frac{f(A)}{2}
	\enspace.
\]
Combining the two last inequalities gives:
\[
	\bE[f(T)]
	\geq
	\frac{1}{2} \cdot \sum_{A \subseteq \OPT} \left(\Pr[\OPT \cap S = A] \cdot \frac{f(A)}{2}\right)
	=
	\frac{\E[f(\OPT \cap S)]}{4}
	\geq
	\frac{f(\OPT)}{8}
	\enspace,
\]
where the last inequality follows by Lemma~\ref{lem:sampling}.
\end{proof}
}

\begin{corollary} \label{cor:f_M_value}
$\bE[f(M)] \geq f(\OPT)/8$.
\end{corollary}
\begin{proof}
We prove the corollary for the set $M$ of Algorithm~\ref{alg:online}, which is fine since it has the same distribution as the set $M$ of Algorithm~\ref{alg:simulated}. Observe that the set $L$ of Algorithm~\ref{alg:online} contains every element with probability $1/2$, independently, and $M$ is the result of applying {\Greedy} to this set. Thus, the corollary holds by Lemma~\ref{lem:greedy_property}.
\end{proof}

\subsection{Second Stage} \label{ssc:second_stage}

In the second stage we use the bounds proved in the first stage to get bounds also on $\bE[w(N)]$ and $\bE[f_w(N)]$.

\begin{lemma} \label{lem:m_n_ratio}
$\bE[w(N)] = p \cdot \bE[w(M)]$.
\end{lemma}
\ifbool{shortVersion}{}{
\begin{proof}
Consider an arbitrary element $u \in E$ processed by Algorithm~\ref{alg:simulated}, and let us fix all history up to the point before $u$ is processed. If $M + u \not \in \cI$ or $f(u \mid M) \leq 0$, then there is a zero expected increase in both $w(M)$ and $w(N)$ during the processing of $u$. Otherwise, the expected increase in $w(M)$ is $w(u)/2$, while the expected increase in $w(N)$ is $(p/2) \cdot w(u)$. Hence, if we denote by $\Delta_M$ the expected increase in $w(M)$ and by $\Delta_N$ the expected increase in $w(N)$, then:
\[
	\Delta_N
	=
	\frac{(p/2) \cdot w(u)}{w(u)/2} \cdot \Delta_M
	=
	p \cdot \Delta_M
	\enspace.
\]
The lemma now follows from the linearity of expectation.
\end{proof}
}

\ifbool{shortVersion}{}{
Getting a bound on $f_w(N)$ is somewhat more involved.
}

\begin{lemma} \label{lem:f_w_m_n_ratio}
$\bE[f_w(N)] \geq \frac{f(\varnothing)}{1 + p} + \frac{p(1 - p)}{1 + p} \cdot \bE[f_w(M)]$.
\end{lemma}
\ifbool{shortVersion}{}{
\begin{proof}
Consider the auxiliary function $\Phi(M, N) = f_w(N) - \frac{p}{1 + p} \cdot f_w(M \cup N)$. Observe that $\Phi$ is in fact a random function since it depends on the random vector $w$. Our first objective in this proof is to show that
\begin{equation} \label{eq:phi_bound}
	\bE[\Phi(M, N)] \geq (1 + p)^{-1} \cdot f(\varnothing)
	\enspace.
\end{equation}
For that purpose, we define $\Delta_u$ as the change in $\Phi(M, N)$ when $u$ is processed by Algorithm~\ref{alg:simulated}. More formally, $\Delta_u = \Phi(M'_u, N'_u) - \Phi(M_u, N_u)$. Additionally, let $R_u$ be an event encoding all the random decisions of Algorithm~\ref{alg:simulated} up to the moment before it processes $u$, and let $\cR_u$ be the set of all such events. Then, since, by Property~\ref{prop:w_always_smaller}, $\Phi(\varnothing, \varnothing) = (1 + p)^{-1} \cdot f_w(\varnothing) = (1 + p)^{-1} \cdot f(\varnothing)$:
\[
	\bE[\Phi(M, N)] - (1 + p)^{-1} \cdot f(\varnothing)
	=
	\sum_{u \in E} \bE[\Delta_u]
	=
	\sum_{u \in E} \sum_{R_u \in \cR_u} \left(\Pr[R_u] \cdot \bE[\Delta_u \mid R_u]\right)
	\enspace.
\]
Thus, to prove Inequality~\eqref{eq:phi_bound} it is enough to show $\bE[\Delta_u \mid R_u] \geq 0$ for an arbitrary element $u \in E$ and event $R_u \in \cR_u$. Notice that conditioned on $R_u$, the sets $M_u$ and $N_u$ and the part of the vector $w$ corresponding to $E'_u$ are all deterministic. If $M_u + u \not \in \cI$ or $f(u \mid M_u) < 0$, then we are done since $M_u = M'_u$ and $N_u = N'_u$. Thus, we only need to consider the case $M_u + u \in \cI$ and $f(u \mid M_u) \geq 0$. In this case $u$ is added to $M$ with probability $1/2$ and to $N$ with probability $p/2$. Thus,
\begin{align*}
	\bE[f_w(N'_u) -{} & f_w(N_u) \mid R_u]
	=
	\Pr[u \in N \mid R_u] \cdot \bE[f_w(u \mid N_u) \mid R_u]\\
	\geq{} &
	\Pr[u \in N \mid R_u] \cdot \bE[f_w(u \mid M_u \cup N_u) \mid R_u] \\
	={} &
	\frac{p/2}{1/2 + p/2} \cdot \Pr[u \in M \cup N \mid R_u] \cdot \bE[f_w(u \mid M_u \cup N_u) \mid R_u]\\
	={} &
	\frac{p}{1 + p} \cdot \bE[f_w(M'_u \cup N'_u) - f_w(M_u \cup N_u) \mid R_u]
	\enspace,
\end{align*}
where the inequality holds since $f_w$ is submodular by Property~\ref{prop:properties_preserved}. Rearranging the last inequality yields $\bE[\Delta_u \mid R_u] \geq 0$, which completes the proof of Inequality~\eqref{eq:phi_bound}.

Next, observe that for an element to enter $N$, three things have to happen: first it must hold that $M_u + u \in \cI$ and $f(u \mid M_u) \geq 0$, then the algorithm must randomly decide not to add the element to $M$ and finally the algorithm must randomly decide (with probability $p$) to add the element to $N$. The last decision does not affect the future development of $M$ and $w$, and thus, even conditioned on $M$ and $w$, every element belongs to $N$ with probability at most $p$. To use the last observation, let $g_{w,M}(S) = f_w(M \cup S)$. One can observe that $g_{w,M}$ is non-negative and submodular. Thus, by Lemma~\ref{lem:max_prob_bound}:
\[
	\bE[f_w(M \cup N) \mid M, w]
	=
	\bE[g_{w,M}(N) \mid M, w]
	\geq
	(1 - p) \cdot g_{w,M}(\varnothing)
	=
	(1 - p) \cdot f_w(M)
	\enspace.
\]
By the law of total expectation, the above inequality implies:
\[
	\bE[f_w(M \cup N)]
	\geq
	(1 - p) \cdot \bE[f_w(M)]
	\enspace,
\]
which implies the lemma when combined with Inequality~\eqref{eq:phi_bound}.
\end{proof}
}

\subsection{Third Stage}

In the third stage we use the bounds proved in the first and second stages to get bounds on $\bE[w(Q \cap N)]$ and $\bE[f_w(Q \cap N)]$. For that purpose, let us define $\OPT_w(N)$ as the maximum weight independent set in $N$ with respect to the weight function $w$. The fact that {\Linear} is $\alpha$-competitive for $\cM$ implies the following observation.

\begin{observation} \label{obs:linear_approximation_ratio}
$\bE[w(Q \cap N)] = \bE[w(Q)] \geq \frac{1}{\alpha} \cdot \bE[w(\OPT_w(N \cup N_0))] = \frac{1}{\alpha} \cdot \bE[w(\OPT_w(N))]$.
\end{observation}

Thus, to get a lower bound on $\bE[w(Q \cap N)]$ it suffices to get a lower bound on $\bE[w(\OPT_w(N))]$.

\begin{lemma} \label{lem:opt_n_bound}
$\bE[w(\OPT_w(N))] \geq \frac{p}{1 + p} \cdot \bE[w(M)]$.
\end{lemma}
\begin{proof}
For this proof we need Algorithm~\ref{alg:simulated} to maintain two additional sets $N'$ and $H$. The set $N'$ is the set of elements of $N$ that are not spanned by previous $N$-elements when added to $N$. More formally, the set $N'$ is originally empty. Whenever an element $u$ is added to $N$, it is also added to $N'$ if it is not spanned by previous elements of $N$.
Clearly, the set $N'\subseteq N$ at the end of the procedure is an independent set, and we can,
thus, use $\E[w(N')]$ as a lower bound on $\E[w(\OPT_w(N))]$.\footnote{In fact, $N'$ is a maximum weight independent set in $N$ since the process creating it is equivalent to running greedy on $N$ with the objective function $w$.} Hence, to prove the lemma it is enough to
show $\E[w(N')]\geq \frac{p}{1+p}\E[w(M)]$.

The set $H$ is maintained by the following rules:
\begin{compactitem}
	\item Originally $H$ is empty.
	\item Whenever an element $u$ is added to $N'$, it is also added to $H$ if $H + u \in \cI$.
	\item Whenever an element $u$ is added to $M$, it is also added to $H$. If that addition makes $H$ non-independent, then an arbitrary element $\phi(u) \in H \cap N'$ such that $H - \phi(u) \in \cI$ (such an element exists since $H \setminus N' = M$ is independent) is removed from $H$.
\end{compactitem}

Consider now an arbitrary element $u \in E$ processed by Algorithm~\ref{alg:simulated}, and let us fix all history up to the point before $u$ is processed. Notice that at this point $w(u)$ is no longer a random variable. We are interested in the expected increase of $w(M)$ and $w(N')$ when $u$ is processed. If $M + u \not \in \cI$ or $f(u \mid M) < 0$, then Algorithm~\ref{alg:simulated} does not add $u$ to either $M$ or $N$, and thus, there is a zero increase in both $w(M)$ and $w(N')$. Otherwise, if $N' + u \in \cI$, then the expected increase in $w(M)$ is $w(u)/2$, while the expected increase in $w(N')$ is $(p/2) \cdot w(u)$. Finally, we need to consider the case that $M + u \in \cI$ and $f(u \mid M) \geq 0$ but $N' + u \not \in \cI$. In this case the expected increase in $w(M)$ is still $w(u)/2$, but the expected increase in $w(N')$ is $0$. To fix that, we charge $(p/2) \cdot w(u)$ to the element of $H$ that becomes $\phi(u)$ if $u$ is added to $M$ (regardless of whether $u$ is really added to $M$ or not).

Let $c(u)$ be the amount charged to an element $u$. By the above discussion we clearly have:
\begin{equation} \label{eq:connection_charges}
	\frac{\bE[w(M)]}{\bE[w(N')] + \sum_{u \in E} \bE[c(u)]}
	=
	\frac{1/2}{p/2}
	=
	p^{-1}
	\enspace.
\end{equation}

To complete the proof we upper bound the charge to every element $u \in E$. Let us fix all history up to the point after $u$ is processed. If $u \not \in H \cap N'$ at this point, then, by definition, $c(u) = 0$. Otherwise, from this point on, till $u$ leaves $H \cap N'$, every arriving element $u'$ can have one of two behaviors:
\begin{compactitem}
	\item $u'$ does not cause a charge to be added to $u$.
	\item $u'$ causes $u$ to be charged by $(p/2) \cdot w(u') \leq (p/2) \cdot w(u)$ (the inequality holds because Algorithm~\ref{alg:simulated} considers elements in a non-increasing weights order). In this case, with probability $1/2$, $u'$ is added to $M$ and $u$ is removed from $H \cap N'$ (and therefore, will not be charged again in the future).
\end{compactitem}
From the above two options we learn that when $u$ gets into $H \cap N'$, $c(u)$ is upper bounded by $(p/2) \cdot w(u) \cdot X$, where $X$ is distributed according to the geometric distribution $G(1/2)$. Thus, in this case:
\[
	\bE[c(u)]
	\leq
	(p/2) \cdot w(u) \cdot \bE[X]
	=
	p \cdot w(u)
	\enspace.
\]
Combining both cases, we get that the following inequality always holds:
\[
	\bE[c(u)]
	\leq
	p \cdot \bE[w(N' \cap \{u\})]
	\enspace.
\]

Plugging the last inequality into~\eqref{eq:connection_charges} gives:
\[
	\frac{\bE[w(M)]}{\bE[w(N')] + p \cdot \bE[w(N')]}
	\leq
	\frac{1}{p}
	\Rightarrow
	\frac{p}{1 + p} \cdot \bE[w(M)]
	\leq
	\bE[w(N')]
	\enspace.
	\qedhere
\]
\end{proof}

\begin{corollary} \label{cor:out_weight}
$\bE[w(Q \cap N)] \geq \frac{p}{\alpha(1 + p)} \cdot \bE[w(M)]$.
\end{corollary}
\begin{proof}
Follows immediately from Observation~\ref{obs:linear_approximation_ratio} and Lemma~\ref{lem:opt_n_bound}.
\end{proof}

The bound on $\bE[f_w(Q \cap N)]$ is obtained by showing that
$\bE[w(Q \cap N) - f_w(Q \cap N)] \leq q \cdot \bE[w(N) - f_w(N)]$ for some value $q \geq \frac{1}{\alpha}$.
We later show that this inequality always holds for $q=1$, which already allows us to
prove Theorem~\ref{thm:main_nonmonotone}. However, it turns out that by exploiting some basic
properties of many algorithms for (linear) \SP, the same inequality can be shown for
smaller values of $q$, which leads to stronger competitive ratios.

\begin{proposition} \label{prop:general_ratio_nonmonotone}
If $\bE[w(Q \cap N) - f_w(Q \cap N)] \leq q \cdot \bE[w(N) - f_w(N)]$ for some value $q \geq \frac{1}{\alpha}$,
then $\bE[f(Q \cap N)] \geq \bE[f_w(Q \cap N)] \geq \frac{p(1 - 2pq\alpha)}{8\alpha (1 + p)} \cdot f(\OPT)$.
\end{proposition}
\ifbool{shortVersion}{}{
\begin{proof}
Observe that:
\begin{align} \label{eqn:first_bound}
	\bE[f_w(Q \cap N)]
	\geq{} &
	\bE[w(Q \cap N)] - q \cdot \bE[w(N) - f_w(N)]\\ \nonumber
	\geq{} &
	\frac{p}{\alpha(1 + p)} \cdot \bE[w(M)] - q \cdot \bE[w(N)] + q \cdot \bE[f_w(N)]\\ \nonumber
	\geq{} &
	\frac{p}{\alpha(1 + p)} \cdot \bE[w(M)] - qp \cdot \bE[w(M)] + \frac{q \cdot f(\varnothing)}{1 + p} + \frac{qp(1 - p)}{1 + p} \cdot \bE[f_w(M)]
	\enspace,
\end{align}
where the first inequality holds by the assumption of the proposition, the second by Corollary~\ref{cor:out_weight} and the last by Lemmata~\ref{lem:m_n_ratio} and~\ref{lem:f_w_m_n_ratio}.

Lemma~\ref{lem:M_equalities} implies the following inequalities:
\begin{gather*}
	\frac{p}{\alpha(1 + p)} \cdot \bE[w(M)]
	=
	\frac{p}{\alpha(1 + p)} \cdot (\bE[f(M)] - f(\varnothing))
	\geq
	\frac{p}{\alpha(1 + p)} \cdot \bE[f(M)] - \frac{q \cdot f(\varnothing)}{1 + p}
	\enspace,\\
	- \bE[w(M)] = -(\bE[f(M)] - f(\varnothing)) \geq -\bE[f(M)]
	\qquad
	\text{and}
	\qquad
	\bE[f_w(M)]
	=
	\bE[f(M)]
	\enspace.
\end{gather*}
Plugging these inequalities into Inequality~\eqref{eqn:first_bound} yields:
\begin{align*}
	\bE[f_w(Q \cap N)]
	\geq{} &
	\frac{p}{\alpha(1 + p)} \cdot \bE[f(M)] - qp \cdot \bE[f(M)] + \frac{qp(1 - p)}{1 + p} \cdot \bE[f(M)]\\
	={} &
	\frac{p(1 - \alpha q(1 + p) + \alpha q(1 - p))}{\alpha(1 + p)} \cdot \bE[f(M)]
	=
	\frac{p(1 - 2qp\alpha)}{\alpha (1 + p)} \cdot \bE[f(M)]
	\enspace.
\end{align*}
The proposition now follows by combining the last inequality with Corollary~\ref{cor:f_M_value}.
\end{proof}
}

The following observation proves that one can use $q = 1$ for every algorithm {\Linear}.

\begin{observation} \label{obs:diff_monotone}
$w(S) - f_w(S)$ is a monotone function of $S$, hence, $w(Q \cap N) - f_w(Q \cap N) \leq w(N) - f_w(N)$.
\end{observation}
\ifbool{shortVersion}{}{
\begin{proof}
Let $S$ be an arbitrary subset of $E$ and let $u \in E \setminus S$. We need to show that $w(S + u) - f_w(S + u) \geq w(S) - f_w(S)$, or equivalently, $w(u) \geq f_w(S + u) - f_w(S)$. By definition, there exists a set $B \subseteq S$ such that:
\[
	f_w(S) = f(B) + w(S \setminus B)
	\enspace.
\]
Hence,
\begin{align*}
	f_w(S + u)
	={} &
	\min_{A \subseteq S + u} \{f(A) + w((S + u)\setminus A)\}
	\leq
	f(B) + w((S + u) \setminus B)\\
	={} &
	f(B) + w(S \setminus B) + w(u)
	=
	f_w(S) + w(u)
	\enspace.
\end{align*}
The observation now follows by rearranging the above inequality.
\end{proof}
}

We can now prove Theorem~\ref{thm:main_nonmonotone}.

\begin{proof}[Proof of Theorem~\ref{thm:main_nonmonotone}]
By Observation~\ref{obs:diff_monotone}, we can plug $q = 1$ into Proposition~\ref{prop:general_ratio_nonmonotone}.
Choosing $p = (3\alpha)^{-1} < 1$, the proposition implies:
\[
	\E[f(Q \cap N)]
	\geq
	\frac{(3\alpha)^{-1}(1 - 2/3)}{8\alpha(1 + (3\alpha)^{-1})} \cdot f(\OPT)
	=
	\frac{1}{24\alpha(3\alpha + 1)} \cdot f(\OPT)
	\enspace.
\]
Hence, Algorithm~\ref{alg:simulated} with $p = (3\alpha)^{-1}$ is $24\alpha(3\alpha + 1)$-competitive for the matroid $\cM$.
\end{proof}

Note that $p=(3\alpha)^{-1}$ used in the above
proof is not the minimizer
of the expression $p(1-2p\alpha)/(8\alpha(1+p))$
obtained from Proposition~\ref{prop:general_ratio_nonmonotone}
by setting $q=1$. We use $p=(3\alpha)^{-1}$ for
simplicity since it leads to a clean expression that
is very close to the one obtained by the minimizing $p$.

\ifbool{shortVersion}{For proofs on the strengthenings of
Theorem~\ref{thm:main_nonmonotone} which hold in some
cases as stated in the introduction, we refer to the
long version of the paper.
}{
In some cases it is possible to use a smaller value of $q$ in Proposition~\ref{prop:general_ratio_nonmonotone}. The following two claims prove one such case. A set function $g\colon 2^E \to \bR$ is called \emph{normalized} if $g(\varnothing)= 0$ and \emph{supermodular} if $g(A) + g(B) \leq g(A \cup B) + g(A \cap B)$.

\begin{lemma} \label{lem:sample_lower_bound}
  Let $g\colon 2^E \rightarrow \bR^+$ be a normalized, monotone and \emph{supermodular} function. Denote by $A(q)$ a random subset of $A$ where each element appears with probability at most $q$ (not necessarily independently). Then, $\mathbb{E}[g(A(q))] \leq q \cdot g(A)$.
\end{lemma}
\begin{proof}
Let $A=\{u_1,\dots, u_{|A|}\}$ be an arbitrary numbering
of the elements in $A$, and for $i\in \{0,\dots, |A|\}$
we define $A_i=\{u_1,\dots, u_i\}$, where $A_0 = \varnothing$.
Denote by $X_i$ an indicator for the event that
$u_i \in A(q)$, and let $q_i=\Pr[u_i \in A(q)] = \mathbb{E}[X_i] \leq q$.
Then:
\begin{align*}
    \bE[g(A(q))]
    ={} &
    \mathbb{E}\left[g(\varnothing) + \sum_{i = 1}^{|A|} X_i \cdot g(u_i \mid A_{i - 1} \cap A(q)) \right]
    \leq
    \mathbb{E}\left[g(\varnothing) + \sum_{i = 1}^{|A|} X_i \cdot g(u_i \mid A_{i - 1}) \right]\\
    ={} &
    g(\varnothing) + \sum_{i = 1}^{|A|} q_i \cdot g(u_i \mid A_{i - 1})
    \leq{}
     g(\varnothing) + q \cdot \sum_{i = 1}^{|A|} g(u_i \mid A_i)\\
    ={}& (1-q)\cdot g(\varnothing) + q \cdot g(A) = q \cdot g(A) \enspace,
\end{align*}
where there first inequality follows from the supermodularity of $g$,
the second one from monotonicity and the fact that $q_i\leq q$ for $1 \leq i \leq |A|$, and the last
equality follows by the fact that $g$ is normalized.
\end{proof}

\begin{corollary} \label{cor:diff_bound}
If {\Linear} is guaranteed to pick every element of its input with probability at most $q$, then $\bE[w(Q \cap N) - f_w(Q \cap N)] \leq q \cdot \bE[w(N) - f_w(N)]$.
\end{corollary}
\begin{proof}
Let us define the function $g(S) = w(S) + f_w(\varnothing) - f_w(S)$. Since $w(S)$ is linear, $f_w(\varnothing)$ is a constant and $f_w(S)$ is submodular, $g(S)$ is a supermodular function. Additionally, $g(S)$ is normalized by construction and monotone by Observation~\ref{obs:diff_monotone}. Hence, for every fixed $N$, by Lemma~\ref{lem:sample_lower_bound} and the assumption that no element of $N$ belongs to $Q$ with probability larger than $q$:
\[
	\bE[g(Q \cap N)] \leq q \cdot g(N)
	\enspace.
\]
Since the last inequality holds for every fixed $N$, it holds also, in expectation, without fixing $N$. Thus:
\begin{align*}
	\bE[w(Q \cap N) - f_w(Q \cap N)]
	={} &
	\bE[g(Q \cap N)] - f_w(\varnothing)\\
	\leq{} &
	q \cdot \bE[g(N)] - f_w(\varnothing)
	\leq
	q \cdot \bE[w(N) - f_w(N)]
	\enspace,
\end{align*}
where the last inequality holds since $f_w$ is non-negative by Property~\ref{prop:properties_preserved}.
\end{proof}

We can now prove Theorem~\ref{thm:bound_probability_nonmonotone}.

\begin{proof}[Proof of Theorem~\ref{thm:bound_probability_nonmonotone}]
Observe that $q$ must be at least $1/\alpha$ since any $\alpha$-competitive algorithm for {\SP} must be able to select an element with probability at least $1/\alpha$ when this element is the only element having a non-zero weight. Hence, by Corollary~\ref{cor:diff_bound}, we can plug $q$ into Proposition~\ref{prop:general_ratio_nonmonotone}. Letting $\beta = \alpha q \geq 1$ and choosing $p = (3\beta)^{-1} < 1$, the proposition implies:
\[
	w(Q \cap N)
	\geq
  \frac{\frac{1}{3\beta} (1-2q\alpha\frac{1}{3\beta})}{8\alpha(1+\frac{1}{3\beta})}
	=
	\frac{\frac{1}{3\beta}(1 - 2/3)}{8\frac{\beta}{q} (1+\frac{1}{3\beta}) } \cdot f(\OPT)
	=
	\frac{q}{24\beta (3\beta + 1)} \cdot f(\OPT)
	\enspace.
\]
Hence, Algorithm~\ref{alg:simulated} with $p = (3\beta^{-1} = (3\alpha q)^{-1}$ has a competitive ratio of at most:
\[
	\frac{24 \beta(3\beta + 1)}{q} 
	=
  \frac{24 \alpha q (3\alpha q + 1)}{q}
	=
	O(q\cdot \alpha^2 )
\]
for the matroid $\cM$.
\end{proof}

\subsection{Proof of Theorem~\ref{thm:union_nonmonotone}} \label{ssc:union_analysis}

In this section we assume $\cM$ is a matroid for which Theorem~\ref{thm:union_nonmonotone} is meaniningful. More specifically, the behavior of {\Linear} on $\cM$ can be characterized as follows.
\begin{compactitem}
	\item For every {\SP} instance over the matroid $\cM$, there exists a set of $k$ (correlated) random sets $\{P_i\}_{i = 1}^k$ such that each set $P_i$ is always independent in $\cM$ and $\bE[w(\bigcup_{i = 1}^k P_i)] \geq w(\OPT_w)/\alpha$, where $w$ is the weight function of the {\SP} instance and $\OPT_w$ is the maximum weight independent set given $w$.
	\item {\Linear} outputs a uniformly random set from $\{P_i\}_{i = 1}^k$.
\end{compactitem}

For every $1 \leq i \leq k$, let $Q_i$ be the set $P_i$ corresponding to the execution of {\Linear}
within Algorithm~\ref{alg:simulated} with $p = (3 \alpha)^{-1}$.

\begin{lemma} \label{lem:union_bound}
$\bE[f(\bigcup_{i=1}^k Q_i \cap N)] \geq \frac{1}{24\alpha(3\alpha + 1)} \cdot f(\OPT)$.
\end{lemma}
\begin{proof}
One can verify that the fact that $Q$ is produced by {\Linear} was used in the above analysis of the competitive ratio of Algorithm~\ref{alg:simulated} only to justify the inequality: $\bE[w(Q)] \geq \frac{1}{\alpha} \cdot \bE[w(\OPT_w(N \cup N_0))]$. Thus, the proof of Theorem~\ref{thm:main_nonmonotone} can be viewed as showing that after executing Algorithm~\ref{alg:simulated} with $p = (3\alpha)^{-1}$ every random subset $S$ of $N \cup N_0$ obeying this inequality must obey also:
\[
	\bE[f(S \cap N)] \geq \frac{1}{24\alpha(3\alpha + 1)} \cdot f(\OPT)
	\enspace.
\]
The lemma now follows since $\bigcup_{i=1}^k Q_i$ is a random subset of $N \cup N_0$ obeying the required inequality.
\end{proof}

The next lemma is necessary to transform the bound given by the last lemma into a bound on the values of the separate sets $Q_i$.

\begin{lemma} \label{lem:sub_additive}
Let $g:2^E \rightarrow \bR^+$ be a non-negative submodular function and let $S_1,\dots, S_k \subseteq E$. Then,
\begin{equation*}
\sum_{i=1}^k f(S_i) \geq f\left(\bigcup_{i=1}^k S_i\right) \enspace.
\end{equation*}
\end{lemma}
\begin{proof}
We prove the result by induction on $k$. It clearly
holds for $k=1$. For $k>1$:
\begin{align*}
	\sum_{i=1}^k f(S_i)
	={} &
  f(S_1) + f(S_2) + \sum_{i=3}^k f(S_i)
  \geq
	f(S_1\cup S_2) + f(S_1\cap S_2) + \sum_{i=3}^k f(S_i)\\
  \geq{} &
	f(S_1\cup S_2) + \sum_{i=3}^k f(S_i)
  \geq
	f\left(\bigcup_{i=1}^k S_i\right)
	\enspace,
\end{align*}
where the first inequality holds by the submodularity of $f$, the second by the non-negativity of $f$ and the last by the induction hypothesis.
\end{proof}

We are now ready to prove Theorem~\ref{thm:union_nonmonotone}.

\begin{proof}[Proof of Theorem~\ref{thm:union_nonmonotone}]
Combining Lemmata~\ref{lem:union_bound} and~\ref{lem:sub_additive}, we get:
\[
	\sum_{i=1}^k \E[f(Q_i \cap N)]
	\geq
	\bE\left[f\left(\bigcup_{i=1}^k Q_i \cap N\right)\right]
	\geq
	\frac{1}{24\alpha(3\alpha + 1)} \cdot f(\OPT)
	\enspace.
\]
Since the output set $Q \cap N$ of Algorithm~\ref{alg:simulated} is a uniformly random subset from $\{Q_i \cap N\}_{i = 1}^k$, we get:
\[
	\bE[f(Q \cap N)]
	=
	\frac{\sum_{i=1}^k \E[f(Q_i \cap N)]}{k}
	\geq
	\frac{1}{24 k \alpha (3\alpha + 1)} \cdot f(\OPT)
	\enspace.
\]
Hence, Algorithm~\ref{alg:simulated} with $p = (3\alpha)^{-1}$ is $24k\alpha(3\alpha+1)$-competitive for the matroid $\cM$.
\end{proof}

}

\section*{Acknowledgement}

We are very grateful to Ola Svensson for many interesting discussions
and comments related to this paper.

\bibliographystyle{plain}

\ifbool{shortVersion}{}{
\appendix

\section{Monotone Functions} \label{app:monotone}

In this section we prove the results for {\MSSP} that apply
to general matroids, \ie, Theorem~\ref{thm:monotone}
and Theorem~\ref{thm:opt_probability_monotone}.
The online algorithm we use to prove these results is given as Algorithm~\ref{alg:online_monotone}, which is a close variant of Algorithm~\ref{alg:online}. Just like Algorithm~\ref{alg:online}, this algorithm has a probability parameter $p$, but the role of this parameter in the algorithm is somewhat different.

\begin{algorithm}[ht]
\caption{\textsf{Monotone-Online}$(p)$} \label{alg:online_monotone}
\tcp{Learning Phase}
Choose $X$ from the binomial distribution $B(n, p)$.\\
Observe (and reject) the first $X$ elements of the input. Let $L$ be the set of these elements.\\

\BlankLine

\tcp{Selection Phase}
Let $M$ be the output of {\Greedy} on the set $L$.\\
Let $N \gets \varnothing$.\\
\For{each arriving element $u \in E \setminus L$}
{
	Let $w(u) \gets 0$.\\
	\If{$u$ is accepted by {\Greedy} when applied to $M + u$}
	{
		Add $u$ to $N$.\\
		Let $M_u \subseteq M$ be the solution of {\Greedy} immediately before it adds $u$ to it.\\
		Update $w(u) \gets f(u \mid M_u)$.\\
	}
	Pass $u$ to {\Linear} with weight $w(u)$.
}
\Return{$Q \cap N$, where $Q$ is the output of {\Linear}}.
\end{algorithm}

Like in Algorithm~\ref{alg:online}, it is easy to show that $L$ contains every element of $E$ with probability $p$, independently. For the analysis of Algorithm~\ref{alg:online_monotone}, we again need an equivalent offline algorithm given as Algorithm~\ref{alg:simulated_monotone}, which is a close variant of Algorithm~\ref{alg:simulated}

\begin{algorithm}[ht]
\caption{\textsf{Monotone-Simulated}$(p)$} \label{alg:simulated_monotone}
\DontPrintSemicolon
\tcp{Initialization}
Let $M, N, N_0 \gets \varnothing$.\\
Let $E' \gets E$.\\

\BlankLine

\tcp{Main Loop}
\While{$E' \neq \varnothing$}
{
	Let $u$ be the element of $E'$ maximizing $f(u \mid M)$, and remove $u$ from $E'$.\\
	Let $w(u) \gets f(u \mid M)$.\\
	\If{$M + u$ is independent in $\cM$}
	{
		\lWith{probability $p$}
		{
			Add $u$ to $M$.
		}
		\lOtherwise
		{
			Add $u$ to $N$.
		}
	}
	\Else
	{
		Let $w(u) \gets 0$.\\
		Add $u$ to $N_0$ with probability $1 - p$.
	}
}
Run {\Linear} with $N \cup N_0$ as the input (in a uniformly random order) and the weights defined by $w$.\\
\Return{$Q \cap N$, where $Q$ is the output of {\Linear}}.
\end{algorithm}

A similar proof to the one given in Appendix~\ref{app:equivalence} can be used to show that Algorithms~\ref{alg:online_monotone} and~\ref{alg:simulated_monotone} share the same joint distribution of the set $M$ and the output set.

\subsection{Analysis of Algorithms~\ref{alg:online_monotone} and~\ref{alg:simulated_monotone}} \label{ssc:analysis_monotone}

Throughout this section we fix an arbitrary matroid $\cM = (E, \cI)$ for which {\Linear} is $\alpha$-competitive and analyze the approximation ratio of Algorithm~\ref{alg:simulated_monotone}. Since Algorithms~\ref{alg:online_monotone} and~\ref{alg:simulated_monotone} share their outputs distribution, the approximation ratio we prove for Algorithm~\ref{alg:simulated_monotone} implies an identical competitive ratio for Algorithm~\ref{alg:online_monotone}. The analysis closely follows the proof given in Section~\ref{sec:analysis}, and we mostly explain in this section how to modify the proof of Section~\ref{sec:analysis} to fit Algorithm~\ref{alg:simulated_monotone}.

Let us start with the first stage in which we bound $\bE[w(M)]$ and $\bE[f_w(M)]$. Lemma~\ref{lem:M_equalities}, which ties both values to $\bE[f(M)]$, still holds, and thus, we only need to bound $\bE[f(M)]$.

\begin{lemma} \label{lem:f_M_value_monotone}
$\bE[f(M)] \geq (p/2) \cdot f(\OPT)$.
\end{lemma}
\begin{proof}
We prove the corollary for the set $M$ of Algorithm~\ref{alg:online_monotone}, which is fine since it has the same distribution as the set $M$ of Algorithm~\ref{alg:simulated_monotone}. Algorithm~\ref{alg:online_monotone} calculates $M$ by applying {\Greedy} to $L$. {\Greedy} is known to have an approximation ratio of $1/2$ for the problem of maximizing a non-negative monotone submodular function subject to a matroid constraint~\cite{FNW78}. Hence, $f(M) \geq f(\OPT(L))/2$, where $\OPT(L)$ is the independent subset of $L$ maximizing $f$.

On the other hand, $L$ contains every element of $E$ with probability $p$, independently. Hence,
\[
	f(M)
	\geq
	\frac{\bE[f(\OPT(L))]}{2}
	\geq
	\frac{\bE[f(\OPT \cap L)]}{2}
	\geq
	(p / 2) \cdot f(\OPT)
	\enspace,
\]
where the second inequality holds since $\OPT \cap L$ is an independent set of $\cM$ and the last inequality holds by Lemma~\ref{lem:sampling}.
\end{proof}

This completes the first stage of the proof. In the second stage we bound $\bE[w(N)]$ and $\bE[f_w(N)]$. In Section~\ref{sec:analysis} the bound on $\bE[w(N)]$ is given by Lemma~\ref{lem:m_n_ratio}. The proof of this lemma uses the probability that an element $u$ processed by Algorithm~\ref{alg:simulated} is added to $M$ (respectively, $N$) given that $M + u \in \cI$ and $f(u \mid M) \geq 0$. Let us denote this probability by $p_M$ (respectively, $p_N$). Using this terminology the proof actually shows $\bE[w(N)] = (p_N/p_M) \cdot \bE[w(M)]$. It can be verified that the proof holds also for Algorithm~\ref{alg:simulated_monotone}, as long as one uses the values $p_M$ and $p_N$ corresponding to this algorithm (which are $p$ and $1 - p$, respectively). Hence, we get the following lemma.

\begin{lemma} \label{lem:m_n_ratio_monotone}
$\bE[w(N)] = \frac{1 - p}{p} \cdot \bE[w(M)]$.
\end{lemma}

Similarly, using the potential function $\Phi(M, N) = f_w(N) - \frac{p_N}{p_M + p_N} \cdot f_w(M \cup N)$, the proof of Lemma~\ref{lem:f_w_m_n_ratio} can be used to show that $\bE[f_w(N)] \geq \frac{p_M}{p_M + p_N} \cdot f(\varnothing) + \frac{p_N}{p_M + p_N} \cdot \bE[f_w(M \cup N)]$ in both Algorithms~\ref{alg:simulated} and~\ref{alg:simulated_monotone}. This gives the following lemma.

\begin{lemma} \label{lem:f_w_m_n_ratio_monotone}
$\bE[f_w(N)] \geq p \cdot f(\varnothing) + (1 - p) \cdot \bE[f_w(M)]$.
\end{lemma}
\begin{proof}
By the above discussion:
\begin{align*}
	\bE[f_w(N)]
	\geq{} &
	\frac{p}{p + (1 - p)} \cdot f(\varnothing) + \frac{1 - p}{p + (1 - p)} \cdot \bE[f_w(M \cup N)]\\
	={} &
	p \cdot f(\varnothing) + (1 - p) \cdot \bE[f_w(M \cup N)]
	\enspace.
\end{align*}
The lemma now follows since $f_w$ is monotone by Property~\ref{prop:properties_preserved}.
\end{proof}

The last lemma bounds $\bE[f_w(N)]$, and thus, completes the second stage of the proof. In the third stage we bound $\bE[w(Q \cap N)]$ and $\bE[f_w(Q \cap N)]$. Observation~\ref{obs:linear_approximation_ratio} bounds $\bE[w(Q \cap N)]$ in terms of $\bE[w(\OPT_w(N))]$, and this observation holds also for Algorithm~\ref{alg:simulated_monotone}. Hence, to get a bound on $\bE[w(Q \cap N)]$ we first need a bound on $\bE[w(\OPT_w(N))]$. The following lemma corresponds to Lemma~\ref{lem:opt_n_bound}.

\begin{lemma} \label{lem:opt_n_bound_monotone}
$\bE[w(\OPT_w(N))] \geq (1 - p) \cdot \bE[w(M)]$.
\end{lemma}
\begin{proof}
The proof of Lemma~\ref{lem:opt_n_bound} shows that there exist a random independent set $N' \subseteq N$ and random variables $\{c(u)\}_{u \in E}$ such that:
\[
	\frac{\bE[w(M)]}{\bE[w(N')] + \sum_{u \in E} \bE[c(u)]}
	=
	\frac{p_M}{p_N}
	\enspace,
	\qquad
	\text{and}
	\qquad
	\bE[c(u)]
	\leq
	\frac{p_N}{p_M} \cdot \bE[w(N' \cap \{u\})]
	\enspace.
\]
It can be verified that this proof holds also for Algorithm~\ref{alg:simulated_monotone}, and thus, plugging in the corresponding values of $p_M$ and $p_N$, we get:
\[
	\frac{\bE[w(M)]}{\bE[w(N')] + \sum_{u \in E} \bE[c(u)]}
	=
	\frac{p}{1 - p}
	\enspace,
	\qquad
	\text{and}
	\qquad
	\bE[c(u)]
	\leq
	\frac{1 - p}{p} \cdot \bE[w(N' \cap \{u\})]
	\enspace.
\]
Combining the two equations gives:
\[
	\frac{\bE[w(M)]}{\bE[w(N')] +\frac{1 - p}{p} \cdot \bE[w(N')]}
	\leq
	\frac{p}{1 - p}
	\Rightarrow
	(1 - p) \cdot \bE[w(M)]
	\leq
	\bE[w(N')]
	\enspace.
\]
The lemma now follows by observing that $N'$ is a possible candidate to be $\OPT_w(N)$, and therefore, $w(\OPT_w(N)) \geq w(N')$.
\end{proof}

\begin{corollary} \label{cor:out_weight_monotone}
$\bE[w(Q \cap N)] \geq \frac{1-p}{\alpha} \cdot \bE[w(M)]$.
\end{corollary}
\begin{proof}
Follows immediately from Observation~\ref{obs:linear_approximation_ratio} and Lemma~\ref{lem:opt_n_bound_monotone}.
\end{proof}

The bound on $\bE[f_w(Q \cap N)]$ is obtained from the bound on $\bE[w(Q \cap N)]$ by showing that $\bE[w(Q \cap N) - f_w(Q \cap N)] \leq q \cdot \bE[w(N) - f_w(N)]$ for some value $q \geq 1/\alpha$. However, like in Section~\ref{sec:analysis}, we first prove the approximation ratio induced by every value of $q$ and only then discuss the values that $q$ can take.

\begin{proposition} \label{prop:general_ratio_monotone}
If $\bE[w(Q \cap N) - f_w(Q \cap N)] \leq q \cdot \bE[w(N) - f_w(N)]$ for some value $q \geq 1/\alpha$, then $\bE[f(Q \cap N)] \geq \bE[f_w(Q \cap N)] \geq \frac{q(1 - p)(p/(\alpha q) - 1 + p)}{2} \cdot f(\OPT)$.
\end{proposition}
\begin{proof}
Observe that:
\begin{align} \label{eqn:first_bound_monotone}
	\bE[f_w(Q \cap N)]
	\geq{} &
	\bE[w(Q \cap N)] - q \cdot \bE[w(N) - f_w(N)]\\ \nonumber
	\geq{} &
	\frac{1-p}{\alpha} \cdot \bE[w(M)] - q \cdot \bE[w(N)] + q \cdot \bE[f_w(N)]\\ \nonumber
	\geq{} &
	\frac{1-p}{\alpha} \cdot \bE[w(M)] - \frac{q(1 - p)}{p} \cdot \bE[w(M)] + q(1 - p) \cdot \bE[f_w(M)]
	\enspace,
\end{align}
where the first inequality holds by the assumption of the proposition, the second by Corollary~\ref{cor:out_weight_monotone} and the last by Lemmata~\ref{lem:m_n_ratio_monotone} and~\ref{lem:f_w_m_n_ratio_monotone} and the non-negativity of $f$.

Lemma~\ref{lem:M_equalities} implies the following inequalities:
\begin{gather*}
	\frac{1 - p}{\alpha} \cdot \bE[w(M)]
	=
	\frac{1 - p}{\alpha} \cdot (\bE[f(M)] - f(\varnothing))
	\geq
	\frac{1-p}{\alpha} \cdot \bE[f(M)] - q(1 - p) \cdot f(\varnothing)
	\enspace,\\
	- \frac{q(1 - p)}{p} \cdot \bE[w(M)] = -\frac{q(1 - p)}{p} \cdot (\bE[f(M)] - f(\varnothing)) \geq -\frac{q(1 - p)}{p} \cdot \bE[f(M)] + q(1 - p) \cdot f(\varnothing)\\
	\text{and}\\
	q(1 - p) \cdot \bE[f_w(M)]
	=
	q(1 - p) \cdot \bE[f(M)]
	\enspace.
\end{gather*}
Plugging these inequalities into Inequality~\eqref{eqn:first_bound_monotone} yields:
\begin{align*}
	\bE[f_w(Q \cap N)]
	\geq{} &
	\frac{1 - p}{\alpha} \cdot \bE[f(M)] - \frac{q(1 - p)}{p} \cdot \bE[f(M)] + q(1 - p) \cdot \bE[f(M)]\\
	={} &
	\frac{(1 - p)(p/\alpha - q + qp)}{p} \cdot \bE[f(M)]
	=
	\frac{q(1 - p)(p/(\alpha q) - 1 + p)}{p} \cdot \bE[f(M)]
	\enspace.
\end{align*}
The proposition now follows by combining the last inequality with Lemma~\ref{lem:f_M_value_monotone}.
\end{proof}

Recall that Observation~\ref{obs:diff_monotone} showed in Section~\ref{sec:analysis} that one can always choose $q = 1$. This proof of this observation holds also for Algorithm~\ref{alg:simulated_monotone}, and thus, we get the following corollary.

\begin{corollary} \label{cor:approximation_ratio_general}
For $p = \frac{2\alpha + 1}{2(\alpha + 1)}$, the approximation ratio of Algorithm~\ref{alg:simulated_monotone}
is at most $8\alpha(\alpha+1)$.
\end{corollary}
\begin{proof}
By Observation~\ref{obs:diff_monotone}, we can plug $q = 1$ into Proposition~\ref{prop:general_ratio_monotone}.
Hence, for $p = \frac{2\alpha + 1}{2(\alpha + 1)}$, the proposition implies:
\begin{align*}
	w(Q \cap N)
	\geq{} &
	\frac{\left(1 - \frac{2\alpha + 1}{2(\alpha+1)}\right)\left(\frac{2\alpha + 1}{2\alpha (\alpha + 1)}
    - 1 + \frac{2\alpha + 1}{2(\alpha + 1)}\right)}{2} \cdot f(\OPT)\\
	={} &
	\frac{\frac{1}{2(\alpha + 1)} \cdot \left(\frac{2 \alpha + 1}{2 \alpha (\alpha + 1)} - \frac{1}{2(\alpha + 1)}\right)}{2} \cdot f(\OPT)
	=
	\frac{1}{8 \alpha (\alpha+1)} \cdot f(\OPT)
	\enspace.
	\qedhere
\end{align*}
\end{proof}

Corollary~\ref{cor:approximation_ratio_general} proves the first part of Theorem~\ref{thm:monotone} corresponding to Theorem~\ref{thm:main_nonmonotone}. The proof can be extended to prove the two other parts of the theorem in the same way this is done in Section~\ref{sec:analysis}.

From this point on we assume {\Linear} has, with respect to the matroid $\cM$, the properties guaranteed by Theorem~\ref{thm:opt_probability_monotone}. In other words, for every {\SP} instance over the matroid $\cM$ there exists a random set $S \subseteq E$ obeying:
\begin{compactitem}
	\item $S$ is always an optimal solution of the {\SP} instance.
	\item For every element $u \in E$, $\Pr[\text{$u$ is selected by Linear}] \geq \frac{1}{\alpha} \cdot \Pr[u \in S]$.
\end{compactitem}

By removing appropriately chosen elements from the output of {\Linear} one can get a new algorithm {\LinearPrime} which still has the above properties, but selects no element with probability larger than $1/\alpha$. Clearly such an algorithm always exists although it might be non-efficient and offline, which is fine since we use {\LinearPrime} only for analysis purposes. The following observation shows that {\LinearPrime} must be $\alpha$-competitive, and thus, all the results proved above hold for it.

\begin{observation}
An algorithm having the properties guaranteed by Theorem~\ref{thm:opt_probability_monotone} with respect to a matroid $\cM$ is $\alpha$-competitive for {\SP} over this matroid.
\end{observation}
\begin{proof}
Let $T$ be the random output set of the algorithm given an {\SP} instance over $M$ with a weight function $w'$. Then:
\[
	\bE[w'(T)]
	=
	\sum_{u \in E} w'(u) \cdot \Pr[u \in T]
	\geq
	\frac{1}{\alpha} \cdot \sum_{u \in E} w'(u) \cdot \Pr[u \in S]
	=
	\frac{1}{\alpha} \cdot \bE[w'(S)]
	\enspace.
\]
The observation now follows since $S$ is always an optimal solution  for the {\SP} instance, and thus, $w'(S)$ is always equal to the value of such an optimal solution.
\end{proof}

We are now ready to prove Theorem~\ref{thm:opt_probability_monotone}.
\begin{proof}[Proof of Theorem~\ref{thm:opt_probability_monotone}]
Let $Q$ and $Q'$ be the sets produced by {\Linear} and {\LinearPrime} when these algorithms are placed in Algorithm~\ref{alg:simulated_monotone}. By construction, $Q \supseteq Q'$, and thus, by the monotonicity of $f$:
\begin{equation} \label{eq:algorithms_connection}
	\bE[f(Q \cap N)] \geq \bE[f(Q' \cap N)]
	\enspace.
\end{equation}

Consider now Corollary~\ref{cor:diff_bound}. It can be verified that this corollary holds also for Algorithm~\ref{alg:simulated_monotone}, and thus, $\bE[w(Q' \cap N) - f_w(Q' \cap N)] \leq \frac{1}{\alpha} \cdot \bE[w(N) - f_w(N)]$. Plugging this inequality into Proposition~\ref{prop:general_ratio_monotone} gives:
\[
	\bE[f(Q' \cap N)]
	\geq
	\frac{(1 - p)(\alpha p/\alpha - 1 + p)}{2\alpha} \cdot f(\OPT)
	=
	\frac{(1 - p)(2p - 1)}{2\alpha} \cdot f(\OPT)
	\enspace.
\]
Choosing $p = 3/4$ now gives:
\[
	\bE[f(Q' \cap N)]
	\geq
	\frac{(1 - 3/4)(3/2 - 1)}{2\alpha} \cdot f(\OPT)
	=
	\frac{1}{16\alpha} \cdot f(\OPT)
	\enspace.
\]
Combing the last inequality with Inequality~\eqref{eq:algorithms_connection} proves that Algorithm~\ref{alg:simulated_monotone} with {\Linear} and $p = 3/4$ is $16\alpha$-competitive for {\MSSP} over $\cM$.
\end{proof}

\section{Results for Specific Matroids} \label{app:specific_matroids}

In this appendix we explain how to get the improved competitive ratios stated in Section~\ref{sec:results} for specific classes of matroids. For transversal matroids and $k$-sparse linear matroids the stated competitive ratios
were obtained through an application of
Theorem~\ref{thm:bound_probability_nonmonotone}.
More precisely, for transversal matroids we
use Theorem~\ref{thm:bound_probability_nonmonotone}
with the \SP\ algorithm
of Korula and P\'al~\cite{KP09},
which is $8$-competitive and picks no element
with probability larger than $1/2$.

For $k$-sparse linear matroids we use Soto's
$k e$-competitive procedure~\cite{S13b}.
To get a stronger result than what we would get through
an application of Theorem~\ref{thm:main_nonmonotone},
we observe that we can
assume that Soto's algorithm selects no
element with probability larger than $q=1/e$.
Soto's algorithm reduces the problem to
the classical matroid secretary problem, 
by losing a factor of $k$.
The classical $e$-competitive algorithm of Dynkin
can easily be adapted such that no element is chosen
with probability more than $1/e$, maintaining
$e$-competitiveness of the algorithm.
More precisely, Dynkin's procedure selects the
heaviest element with some probability $p(n)\geq 1/e$
that only depends on the size
$n=|E|$ of the ground set and can be calculated upfront.
Moreover, no element is selected with a probability
exceeding $p(n)$.
We can now modify Dynkin's algorithm as follows,
to obtain an algorithm that is still $e$-competitive
and selects no element with probability larger than
$1/e$: whenever Dynkin's algorithm would select an
element $u$, we will toss a coin and only select it
with probability $1/(e\cdot p(n))$.
Hence, using Soto's algorithm with this modification
of Dynkin's procedure, we have a $k e$-competitive
algorithm for \SP\ over $k$-sparse linear matroids
that selects no element with probability larger than
$1/e$.
The competitive ratios claimed in
Section~\ref{sec:results} now follow from
Theorem~\ref{thm:bound_probability_nonmonotone}
(for \SSP) and Theorem~\ref{thm:monotone} (for \MSSP).

Unitary partition matroids have a trivial $e$-competitive algorithm applying the classical secretary algorithm to every elements class separately. However, this algorithm has two weaknesses. First, it requires prior knowledge about the number of elements in each class, which we sometimes want to avoid. Second, it might select an element with a probability significantly larger than $1/e$, which prevents an effective use of Theorem~\ref{thm:bound_probability_nonmonotone}. In Section~\ref{ssc:partition_algorithm} we present a simple algorithm avoiding these issues. Plugging this algorithm into Theorem~\ref{thm:bound_probability_nonmonotone} yields the guaranteed competitive ratio for {\SSP} over unitary partition matroids.

Ma et al.~\cite{MTW13} give an algorithm for {\SP} over laminar matroids selecting every element of $\OPT$ with probability at least $\nicefrac{1}{9.6}$ and no element with probability larger than $0.158$. Plugging this algorithm as {\Linear} into Theorems~\ref{thm:bound_probability_nonmonotone} and~\ref{thm:opt_probability_monotone} results in competitive ratios
of $1279$ and $154$ for {\SSP} and {\MSSP} over laminar matroids, respectively.\footnote{The result for {\SSP} can be improved to $1196$ by optimizing the parameter of the algorithm of~\cite{MTW13}.} This already improves over the result of~\cite{MTW13} for {\MSSP} and provides the first $O(1)$-competitive algorithm for {\SSP}. However, by a careful combination of our ideas and a lemma from~\cite{MTW13} one can prove the stronger ratios stated in Section~\ref{sec:results}. The details can be found in Section~\ref{app:laminar}.

\subsection{Algorithm for {\SP} over Unitary Partition Matroids} \label{ssc:partition_algorithm}

In a unitary partition matroid the elements are partitioned into $k$ disjoint non-empty classes $\{G_i\}_{i = 1}^k$. A set $S$ is independent in the matroid if and only if it contains at most one element of every class $G_i$. Algorithm~\ref{alg:partition} is an algorithm for {\SP} over unitary partition matroids. Observe that this algorithm can be implemented with no access to the matroid besides of the independence oracle. As usual, the algorithm assumes all weights are non-negative and disjoint. This assumption is without loss of generality since negative weights can be replaced with zero weights (as long as elements of negative weights are rejected even if the algorithm tries to accept them) and ties between weights can be broken arbitrarily.

\begin{algorithm}[ht]
\caption{\textsf{Partition-{\SP}}$(n)$} \label{alg:partition}
\tcp{Learning Phase}
Let $t \gets \lceil n/e \rceil$.\\
\lWith{probability $t - n/e$}{Let $X \gets t - 1$.}
\lOtherwise{Let $X \gets t$.}
Observe (and reject) the first $X$ elements of the input. Let $L$ be the set of these elements.\\

\BlankLine

\tcp{Selection Phase}
\For{each arriving element $u \in E \setminus L$}
{
	Let $G_u$ be the class of $u$.\\
	\If{$G_u$ has not been marked previously}
	{
		\If{$G_u \cap L \neq \varnothing$}
		{
			\If{$w(u) > \max_{u' \in G_u \cap L} w(u')$\label{ln:weight_comparison}}
			{
				Mark $G_u$.\\
				Accept $u$.\\
			}
		}
		\Else
		{
			Mark $G_u$.\\
			Accept $u$ with probability $X/|L|$.
		}
	}
	Add $u$ to $L$.\\
}
\end{algorithm}

Informally, Algorithm~\ref{alg:partition} observes (and discards) the $X \approx n/e$ first elements. Then, for every arriving element $u$ the algorithms makes the following decisions.
\begin{compactitem}
	\item If $u$'s class has been marked by a previous element, then $u$ is simply discarded. A marked class is a class from which we should accept no new elements. A class is usually marked because some element of it has already been accepted, but not always.
	\item Otherwise, if some elements of $u$'s class have already arrived and $u$ is better than all of them, then $u$ is accepted and its class is marked.
	\item Finally, if no element of $u$'s class has previously arrived, then its class is marked and with some probability $u$ is accepted. This is the only case in which a class can get marked despite the fact that no element of it has been accepted to the solution.
\end{compactitem}

\begin{observation}
Algorithm~\ref{alg:partition} always produces an independent set.
\end{observation}
\begin{proof}
Once an element of a class $G$ is accepted, the class $G$ becomes marked and the algorithm dismisses immediately any future elements of this class.
\end{proof}

For every element $u \in E$, let $L_u$ be the set $L$ before $u$ is processed, \ie, it is the set of elements that arrived before $u$. Additionally, for every class $G_i$, let $u^*_i$ be the element with the largest weight in $G_i$ (which is also the element of $G_i$ in the optimal solution).

\begin{lemma} \label{lem:opt_prob}
For every $1 \leq i \leq k$, $u^*_i$ is accepted with probability $\frac{\lceil n/e \rceil}{n} - \frac{1}{e} + \sum_{j = \lceil n/e \rceil}^{n - 1} \frac{1}{ej}$.
\end{lemma}
\begin{proof}
Fix the variable $X$ and the set $L_{u^*_i}$ (and thus, also the location of $u^*_i$ in the input), and let us calculate the probability $u^*_i$ is accepted. If $|L_{u^*_i}| < X$, then $u^*_i$ arrives during the learning phase (\ie, it is one of the first $X$ elements), and thus, is it not accepted. If $|L_{u^*_i}| \geq X$, then there are three cases.
\begin{compactitem}
	\item The first case is when $G_i \cap L_{u^*_i} = \varnothing$. In this case $u^*_i$ is the first element of $G_i$ encountered by Algorithm~\ref{alg:partition}. Since it is encountered after the learning phase (\ie, it is not one of the first $X$ elements), it is accepted with probability $X/|L_{u^*_i}|$.
	\item The second case is when $G_i \cap L_{u^*_i} \neq \varnothing$ and the element $\hat{u}_i$ with the maximum weight in $G_i \cap L_{u^*_i}$ arrives during the learning phase. In this case every element of $G_i \cap L_{u^*_i}$ arriving after the learning phase fails the condition on Line~\ref{ln:weight_comparison} of Algorithm~\ref{alg:partition}. Thus, $G_i$ is still unmarked when $u^*_i$ arrives, and thus, it is accepted. Since the order of the elements in $L_{u^*_i}$ is uniformly random, this case happens with probability $X / |L_{u^*_i}|$ conditioned on $X$ and $L_{u^*_i}$.
	\item The final case is when $G_i \cap L_{u^*_i} \neq \varnothing$ and the element $\hat{u}_i$ with the maximum weight in $G_i \cap L_{u^*_i}$ arrives after the learning phase. In this case, either $G_i$ is already marked when $\hat{u}_i$ arrives or $G_i$ is marked when $\hat{u}_i$ is processed. Either way, when $u^*_i$ arrives $G_i$ is already marked, and thus, $u^*_i$ is dismissed.
\end{compactitem}
The above cases imply:
\[
	\Pr[\mbox{$u^*_i$ is accepted} \mid X, L_{u^*_i}]
	=
	\begin{cases}
		\frac{X}{|L_{u^*_i}|} & \mbox{if $|L_{u^*_i}| \geq X$} \enspace,\\
		0 & \mbox{otherwise} \enspace.
	\end{cases}
	\enspace.
\]

Regardless of $X$, the arrival time of $u^*_i$ is distributed uniformly between $1$ and $n$. Hence,
\[
	\Pr[\mbox{$u^*_i$ is accepted} \mid X]
	=
	\frac{1}{n} \cdot \sum_{j = X + 1}^n \frac{X}{j - 1}
	=
	\sum_{j = X}^{n - 1} \frac{X}{nj}
	\enspace.
\]
The law of total probability now gives:
\begin{align*}
	\Pr[\mbox{$u^*_i$ is accepted}]
	={} &
	\left(t - \frac{n}{e}\right) \cdot \sum_{j = t - 1}^{n - 1} \frac{t - 1}{nj} + \left(1 - t + \frac{n}{e}\right) \cdot \sum_{j = t}^{n - 1} \frac{t}{nj}\\
	={} &
	\left(\frac{t}{n} - \frac{1}{e}\right) + \left[(t - 1)\left(t - \frac{n}{e}\right) + t\left(1 - t + \frac{n}{e}\right)\right] \cdot \sum_{j = t}^{n - 1} \frac{1}{nj}\\
	={} &
	\frac{t}{n} - \frac{1}{e} + \sum_{j = t}^{n - 1} \frac{1}{ej}
	=
	\frac{\lceil n/e \rceil}{n} - \frac{1}{e} + \sum_{j = \lceil n/e \rceil}^{n - 1} \frac{1}{ej}
	\enspace.
	\qedhere
\end{align*}
\end{proof}

For ease of notation, let us denote $\alpha(n) = \left(\lceil n/e \rceil/n - e^{-1} + \sum_{j = \lceil n/e \rceil}^{n - 1} (ej)^{-1}\right)^{-1}$. Lemma~\ref{lem:opt_prob} shows that every element of $\OPT$ is accepted by Algorithm~\ref{alg:partition} with probability $1/\alpha(n)$, and thus, Algorithm~\ref{alg:partition} is $\alpha(n)$-competitive for {\SP} over unitary partition matroids.

\begin{observation} \label{obs:math_calc}
For every $n \geq 1$, $\alpha(n) \leq e$.
\end{observation}
\begin{proof}
\begin{align*}
	\sum_{j = \lceil n/e \rceil}^{n - 1} \frac{1}{ej}
	\geq{} &
	\frac{1}{e} \cdot \int_{\lceil n/e \rceil}^n \frac{dx}{x}
	=
	\frac{1}{e} \cdot \left[\int_{n/e}^n \frac{dx}{x} - \int_{n/e}^{\lceil n/e \rceil} \frac{dx}{x} \right]\\
	\geq{} &
	\frac{1}{e} \cdot [\ln x]_{\lceil n/e \rceil}^n - \int_{n/e}^{\lceil n/e \rceil} \frac{dx}{n}
	=
	\frac{1}{e} + \frac{n/e - \lceil n/e \rceil}{n}
	=
	\frac{2}{e} - \frac{\lceil n/e \rceil}{n}
	\enspace.
	\qedhere
\end{align*}
\end{proof}

The next lemma completes the analysis of Algorithm~\ref{alg:partition}.
\begin{lemma}
No element is accepted with probability larger than $1/\alpha(n)$ by Algorithm~\ref{alg:partition}.
\end{lemma}
\begin{proof}
Fix an arbitrary element $u$ of a class $G_i$. For every arrival order $\pi$ of the elements, let us define $\sigma(\pi)$ as the order obtained from $\pi$ by replacing $u$ and $u^*_i$. Observe that $\sigma$ is a bijection.

Consider now an arbitrary order $\pi$ given which Algorithm~\ref{alg:partition} accepts $u$ with probability $p > 0$. Clearly $u$ proceeds $u^*_i$ in $\pi$ because otherwise $w(u)$ could not be larger than the weights or all elements appearing earlier in $\pi$, which is a necessary condition for $u$ to be accepted. Thus, Algorithm~\ref{alg:partition} is in the same state when processing $u$ given $\pi$ and $u^*_i$ given $\sigma(\pi)$. Since $w(u^*_i) > w(u)$, this implies that Algorithm~\ref{alg:partition} accepts $u^*_i$ with probability of at least $p$ given the order $\sigma(\pi)$.

Recalling that $\sigma$ is a bijection, we now get:
\begin{align*}
	\Pr[\mbox{$u$ is accepted}]
	={} &
	\sum_{\pi} \frac{\Pr[\mbox{$u$ is accepted given $\pi$}]}{n!}
	\geq
	\sum_{\pi} \frac{\Pr[\mbox{$u^*_i$ is accepted given $\sigma(\pi)$}]}{n!}\\
	={} &
	\sum_{\pi} \frac{\Pr[\mbox{$u^*_i$ is accepted given $\pi$}]}{n!}
	=
	\Pr[\mbox{$u^*_i$ is accepted}]
	=
	\frac{1}{\alpha(n)}
	\enspace.
	\qedhere
\end{align*}
\end{proof}

\noindent \textbf{Remark}: The fact that Algorithm~\ref{alg:partition} is $\alpha(n)$-competitive while accepting no element with a probability larger than $1/\alpha(n)$ is sufficient to
extract all the power of
Theorem~\ref{thm:bound_probability_nonmonotone}
and get the result stated in the table in Section~\ref{sec:results}.
However, sometimes it is useful to have an algorithm that is $e$-competitive and accepts no element with a probability larger than $1/e$. This can be achieved by executing Algorithm~\ref{alg:partition} and accepting with probability $\alpha(n)/e$ every element accepted by Algorithm~\ref{alg:partition}.

\subsection{Results for Laminar Matroids} \label{app:laminar}

In this section we prove the competitive ratios stated in Section~\ref{sec:results} for laminar matroids. Throughout this section we assume $\cM$ is a laminar matroid. Ma et al.~\cite{MTW13} consider an algorithm which constructs the sets $M$ and $N$ in the same way this is done by Algorithm~\ref{alg:simulated_monotone}. For this algorithm, and thus, also for Algorithm~\ref{alg:simulated_monotone}, they proved the following lemma (assuming $\cM$ is laminar).

\begin{lemma}[Rephrased version of Lemma~10 of~\cite{MTW13}] \label{lem:laminar_external}
Consider Algorithm~\ref{alg:simulated_monotone} with 
{\Linear} being the procedure that accepts
every element of $N$ whose acceptance does not violate the
independence of the current solution.
Then 
\[
	\bE[w(Q)]
	\geq
	\left(1 - \frac{2\beta}{(1 - \beta)^3} \right) \cdot \bE[w(N)]
	\enspace,
\]
where $\beta = 2e(1 - p)$.
\end{lemma}

The last lemma can be used instead of Corollary~\ref{cor:out_weight_monotone} in the analysis of Algorithm~\ref{alg:simulated_monotone}. This allows us to prove the result for {\MSSP}.

\begin{corollary}
There exists a $144$-competitive algorithm for {\MSSP} over laminar matroids.
\end{corollary}
\begin{proof}
Consider Algorithm~\ref{alg:simulated_monotone}
using the algorithm {\Linear} defined by Lemma~\ref{lem:laminar_external}.
By Lemma~\ref{lem:laminar_external} we get
\[
	\bE[w(Q \cap N)]
	=
	\bE[w(Q)]
	\geq
	\left(1-\frac{2\beta}{(1-\beta)^3}\right) \cdot \bE[w(N)],
\]
where $\beta=2e(1-p)$.
Thus,
\begin{align} \label{eqn:first_bound_laminar_monotone}
	\bE[f_w(Q \cap N)]
	\geq{} &
	\bE[w(Q \cap N)] - \bE[w(N)] + \bE[f_w(N)]\\ \nonumber
	\geq{} &
	\left(1-\frac{2\beta}{(1-\beta)^3}\right)
    \cdot \bE[w(N)] - \bE[w(N)]
    + (1-p) \cdot \bE[f_w(M)] \\ \nonumber
	={} &
	(1-p) \cdot \bE[f_w(M)] - \frac{2\beta}{(1-\beta)^3} \cdot \bE[w(N)]
	\enspace,
\end{align}
where the first inequality holds by Observation~\ref{obs:diff_monotone} and the second by Lemmata~\ref{lem:f_w_m_n_ratio_monotone} and the non-negativity of $f$.

Recall that, by Lemmata~\ref{lem:M_equalities} and~\ref{lem:m_n_ratio_monotone}, $f_w(M) = f(M)$ and
\[
	\bE[w(N)]
	=
	\frac{1 - p}{p} \cdot \bE[w(M)]
	=
	\frac{1-p}{p} \cdot \bE[f(M) - f(\varnothing)]
	\leq
	\frac{1-p}{p} \cdot \bE[f(M)]
	\enspace.
\]
Plugging these inequalities into Inequality~\eqref{eqn:first_bound_laminar_monotone} yields:
\begin{align*}
	\bE[f_w(Q \cap N)]
	\geq{} &
	(1-p) \cdot \bE[f_w(M)] - \frac{2\beta}{(1-\beta)^3} \cdot \bE[w(N)]\\
	\geq{} &
	(1-p) \cdot \bE[f(M)] - \frac{2\beta}{(1-\beta)^3}\cdot
     \frac{1-p}{p} \cdot \bE[f(M)] \\
  \geq{} & 
  \frac{p}{2}\left(
    1-p - \frac{2\beta}{(1-\beta)^3} \cdot \frac{1-p}{p}
  \right) \cdot \bE[f(\OPT)]\\
  ={} & 
  \frac{p}{2}\left(
    1-p - \frac{4 e (1-p)}{(1-2e(1-p))^3} \cdot \frac{1-p}{p}
  \right) \cdot \bE[f(\OPT)],
\end{align*}
where the last inequality follows from $\bE[f(M)] \geq \frac{p}{2} f(\OPT)$,
as stated by Lemma~\ref{lem:f_M_value_monotone}.
The above expression is maximized for $p\approx 0.976299$ leading
to the claimed competitive ratio.
\end{proof}

To prove the result for {\SSP}, we first need an observation about Algorithm~\ref{alg:simulated}. Consider Algorithm~\ref{alg:simulated_M_N_simplified}.

\begin{algorithm}[ht]
\caption{\textsf{Middle-Algorithm}$(p)$} \label{alg:simulated_M_N_simplified}
\DontPrintSemicolon
\tcp{Initialization}
Let $M, N \gets \varnothing$.\\
Let $E' \gets E$.\\

\BlankLine

\tcp{Main Loop}
\While{$E' \neq \varnothing$}
{
	Let $u$ be the element of $E'$ maximizing $f(u \mid M)$, and remove $u$ from $E'$.\\
	Let $w(u) \gets f(u \mid M)$.\\
	\With{probability $(1 + p)/2$}
	{
		\If{$M + u$ is independent in $\cM$ and $w(u) \geq 0$ \label{ln:acceptance_conditions}}
		{
			\lWith{probability $1/(1 + p)$}
			{
				Add $u$ to $M$.
			}
			\lOtherwiseWith{probability $p / (1 + p)$}
			{
				Add $u$ to $N$.
			}
		}
	}
}
\end{algorithm}

Observe that every processed element obeying the conditions on Line~\ref{ln:acceptance_conditions} of the Algorithm~\ref{alg:simulated_M_N_simplified} is added by both Algorithms~\ref{alg:simulated} and Algorithm~\ref{alg:simulated_M_N_simplified} with probability $1/2$ to $M$ and with probability $p/2$ to $N$. Hence, both algorithms produce the same joint distribution of $M$ and $N$. On the other hand, in terms of the sets $M$ and $N$ Algorithm~\ref{alg:simulated_M_N_simplified} is identical to Algorithm~\ref{alg:simulated_monotone} up to two modifications:
\begin{compactitem}
	\item The value $p$ in Algorithm~\ref{alg:simulated_monotone} is replaced with $1/(1 + p)$.
	\item Algorithm~\ref{alg:simulated_M_N_simplified} dismisses some elements without processing them. These elements include all the elements of negative weights and every other element with probability $1 - (1 + p)/2$.
\end{compactitem}
The second modification does not affect the proof of~\cite{MTW13} for Lemma~\ref{lem:laminar_external} (in fact, the proof requires elements of negative weight to be dismissed). Hence, we get the following lemma which is obtained from Lemma~\ref{lem:laminar_external} by replacing $p$ with $1/(1 + p)$.

\begin{lemma} \label{lem:laminar_nonmonotone}
Consider Algorithm~\ref{alg:simulated} with 
{\Linear} being the procedure that accepts
every element of $N$ whose acceptance does not violate the
independence of the current solution.
Then 
\[
	\bE[w(Q)]
	\geq
	\left(1 - \frac{2\beta}{(1 - \beta)^3} \right) \cdot \bE[w(N)]
	\enspace,
\]
where $\beta = 2e(1 - 1/(1 + p))$.
\end{lemma}

The last lemma can be used instead of Corollary~\ref{cor:out_weight} in the analysis of Algorithm~\ref{alg:simulated}. This allows us to prove the result for {\SSP}.

\begin{corollary}
There exists a $585$-competitive algorithm for {\SSP} over laminar matroids.
\end{corollary}
\begin{proof}
Consider Algorithm~\ref{alg:simulated} using the algorithm {\Linear} defined by Lemma~\ref{lem:laminar_nonmonotone}.
By Lemma~\ref{lem:laminar_nonmonotone} we obtain
\[
	\bE[w(Q \cap N)]
	=
	\bE[w(Q)]
	\geq
	\left(1-\frac{2\beta}{(1-\beta)^3}\right) \cdot \bE[w(N)],
\]
where $\beta=2e(1-1/(1+p))$.
Thus,
\begin{align} \label{eqn:first_bound_laminar_nonmonotone}
	\bE[f_w(Q \cap N)]
	\geq{} &
	\bE[w(Q \cap N)] - \bE[w(N)] + \bE[f_w(N)]\\ \nonumber
	\geq{} &
	\left(1-\frac{2\beta}{(1-\beta)^3}\right)
    \cdot \bE[w(N)] - \bE[w(N)] + \frac{p(1-p)}{1+p}
      \cdot \bE[f_w(M)] \\ \nonumber
	={} &
	\frac{p(1-p)}{1+p} \cdot \bE[f_w(M)] - \frac{2\beta}{(1-\beta)^3} \cdot \bE[w(N)]
	\enspace,
\end{align}
where the first inequality holds by Observation~\ref{obs:diff_monotone} and the second by Lemmata~\ref{lem:f_w_m_n_ratio} and the non-negativity of $f$.

Recall that by Lemmata~\ref{lem:M_equalities} and~\ref{lem:m_n_ratio}, $f_w(M) = f(M)$ and
\[
	\bE[w(N)]
	=
	p \cdot \bE[w(M)]
	=
	p \cdot \bE[f(M) - f(\varnothing)]
	\leq
	p \cdot \bE[f(M)]
	\enspace.
\]
Plugging these inequalities into Inequality~\eqref{eqn:first_bound_laminar_monotone} yields:
\begin{align*}
	\bE[f_w(Q \cap N)]
	\geq{} &
	\frac{p(1-p)}{1+p} \cdot \bE[f_w(M)] - \frac{2\beta}{(1-\beta)^3} \cdot \bE[w(N)]\\
	\geq{} &
	 p \left(\frac{1-p}{1+p} - \frac{2\beta}{(1-\beta)^3}\right) \cdot \bE[f(M)]\\
  \geq{} &
	 \frac{p}{8} \left(\frac{1-p}{1+p} - \frac{2\beta}{(1-\beta)^3}\right) \cdot 
   \bE[f(\OPT)]\\
  ={} &
	 \frac{p}{8} \left(\frac{1-p}{1+p}
     - \frac{4 e (1-1/(1+p))}{(1-2e(1-1/(1+p)))^3}\right) \cdot 
   \bE[f(\OPT)],
\end{align*}
where the last inequality follows by $\E[f(M)]\geq f(\OPT)/8$ which
holds by Corollary~\ref{cor:f_M_value}.
The above expression is maximized for
$p\approx 0.023769$, leading to the claimed competitive ratio.
\end{proof}

\section{Equivalence of Algorithms~\ref{alg:online} and~\ref{alg:simulated}} \label{app:equivalence}

In this section we prove that Algorithms~\ref{alg:online} and~\ref{alg:simulated} share the joint distribution of the set $M$ and the output set. First, for completeness, let us state the standard algorithm {\Greedy} as Algorithm~\ref{alg:greedy}. Recall that {\Greedy} is used by Algorithm~\ref{alg:online}, and observe that, as mentioned in Section~\ref{sec:algorithm}, the solution $M$ of {\Greedy} starts as the empty set and elements are added to it one by one. 

\begin{algorithm}[ht]
\caption{\textsf{Greedy}} \label{alg:greedy}
\tcp{Initialization}
Let $M \gets \varnothing$.\\
Let $E' \gets E$.\\

\BlankLine

\tcp{Main Loop}
\While{$E' \neq \varnothing$}
{
	Let $u$ be the element of $E'$ maximizing $f(u \mid M)$, and remove $u$ from $E'$.\\
	\If{$M + u$ is independent in $\cM$ and $f(u \mid M) \geq 0$}
	{
		Add $u$ to $M$.\\
	}
}
\Return{$M$}.\\
\end{algorithm}

Let $L$ and $F$ be two independent random subsets of $E$, where the set $L$ (respectively, $F$) contains every element $u \in E$ with probability $1/2$ (respectively, $p$), independently. By Observation~\ref{obs:LpropP}, the set $L$ constructed by Algorithm~\ref{alg:online} has exactly the same distribution as the above set $L$. Hence, given access to the above sets $L$ and $F$, Algorithm~\ref{alg:online} can be rewritten as Algorithm~\ref{alg:online_coupled}. Notice that the only changes made during the rewrite are omission of the calculation of $L$ and replacement of the random coin toss on Line~\ref{line:online_change} with a reference to $F$.

\SetKwIF{With}{OtherwiseWith}{Otherwise}{with}{do}{otherwise with}{otherwise}{}
\begin{algorithm}[ht]
\caption{\textsf{Coupled Online}$(L, F)$} \label{alg:online_coupled}
Let $M$ be the output of {\Greedy} on the set $L$.\\
Let $N \gets \varnothing$.\\
\For{each arriving element $u \in E \setminus L$}
{
	Let $w(u) \gets 0$.\\
	\If{$u$ is accepted by {\Greedy} when applied to $M + u$}
	{
		\If{$u \in F$\label{line:online_change}}
		{
			Add $u$ to $N$.\\
			Let $M_u \subseteq M$ be the solution of {\Greedy} immediately before it adds $u$ to it.\\
			Update $w(u) \gets f(u \mid M_u)$.\\
		}
	}
	Pass $u$ to {\Linear} with weight $w(u)$.
}
\Return{$Q \cap N$, where $Q$ is the output of {\Linear}}.
\end{algorithm}

Similarly, Algorithm~\ref{alg:simulated} can be rewritten using $L$ and $F$ as Algorithm~\ref{alg:simulated_coupled}. In this case the coin tosses on Lines~\ref{line:first_change} and~\ref{line:third_change} have been replaced with references to $L$, and the coin toss on Line~\ref{line:second_change} has been replaced with a reference to $F$. It is important to note that an element is never checked at both Lines~\ref{line:first_change} and~\ref{line:third_change}, and thus, the independence of the randomness used by these lines is preserved.

\begin{algorithm}[ht]
\caption{\textsf{Coupled Simulated}$(L, F)$} \label{alg:simulated_coupled}
\DontPrintSemicolon
\tcp{Initialization}
Let $M, N, N_0 \gets \varnothing$.\\
Let $E' \gets E$.\\

\BlankLine

\tcp{Main Loop}
\While{$E' \neq \varnothing$}
{
	Let $u$ be the element of $E'$ maximizing $f(u \mid M)$, and remove $u$ from $E'$.\\
	Let $w(u) \gets f(u \mid M)$.\\
	\If{$M + u$ is independent in $\cM$ and $w(u) \geq 0$}
	{
		\lIf{$u \in L$}
		{
			Add $u$ to $M$.\label{line:first_change}
		}
		\lElseIf{$u \in F$}
		{
			Add $u$ to $N$.\label{line:second_change}
		}
		\Else
		{
			Update $w(u) \gets 0$.\\
			Add $u$ to $N_0$.
		}
	}
	\Else
	{
		Let $w(u) \gets 0$.\\
		\lIf{$u \not \in L$}{Add $u$ to $N_0$.\label{line:third_change}}
	}
}
Run {\Linear} with $N \cup N_0$ as the input (in a uniformly random order) and the weights defined by $w$.\\
\Return{$Q \cap N$, where $Q$ is the output of {\Linear}}.
\end{algorithm}

In the rest of this section we show that
conditioned on $L$ and $F$, Algorithms~\ref{alg:online_coupled} and~\ref{alg:simulated_coupled}
produce the same set $M$ and have the same output distribution. Hence, the
original Algorithm~\ref{alg:simulated} has the same joint distribution of the
set $M$ and the output set as Algorithm~\ref{alg:online}.
Let us start the proof with the following simple observation.

\begin{observation} \label{obs:N_0_elements}
The elements passed to {\Linear} by both Algorithms~\ref{alg:online_coupled} and~\ref{alg:simulated_coupled} are exactly the elements of $E \setminus L$. Moreover, the elements of $E \setminus (L \cup N)$ are all assigned $0$ weights.
\end{observation}

Next, observe that Algorithm~\ref{alg:simulated_coupled} obtains $M$ by running a modified version of {\Greedy} on the elements of $E$. The modified version is identical to the original one, except that every time immediately before adding an element $u \in E \setminus L$ to the result, the element is ``stolen'' (and gets to $N$ or $N_0$ instead). We denote such a modified run of {\Greedy} by {\Greedy}$(E, E \setminus L)$. More generally, we use {\Greedy}$(A, B)$ to denote a run of greedy on the set $A$, where the elements of $B \subseteq A$ are stolen if {\Greedy} tries to add them to its solution.

\begin{observation} \label{obs:greedy_properties}
For four sets $A \supseteq B$ and $C \supseteq D$ obeying $\Greedy(A, B) \subseteq C \setminus D \subseteq A \setminus B$, the runs $\Greedy(A, B)$ and $\Greedy(C, D)$ behave the same in the following sense:
\begin{itemize}
	\item The outputs of the two runs are identical, \ie, $\Greedy(A, B) = \Greedy(C, D)$.
	\item The run $\Greedy(A, B)$ attempts to add an element of $A \cap C$ to the solution if and only if the run $\Greedy(C, D)$ attempts to add this element, moreover, both runs attempt to add the element to the same partial solution.
\end{itemize}
\end{observation}
\begin{proof}
The behavior of {\Greedy} depends only on elements that were added to the result set. Hence, elements that appear in the input but are stolen before being added to the output does not affect the output of {\Greedy}, and the same holds for elements that greedy does not attempt to add to the output.
\end{proof}

\begin{corollary}
Conditioned on a set $L$, Algorithms~\ref{alg:online_coupled} and~\ref{alg:simulated_coupled} produce the same set $M$.
\end{corollary}
\begin{proof}
Algorithm~\ref{alg:online} selects $M$ as $\Greedy(L, \varnothing)$, while Algorithm~\ref{alg:simulated_coupled} select $M$ as $\Greedy(E, E \setminus L)$. The equality between these sets follows immediately from Observation~\ref{obs:greedy_properties} since:
\[
	\Greedy(L, \varnothing) \subseteq L = E \setminus (E \setminus L)
	\enspace.
	\qedhere
\]
\end{proof}

\begin{corollary}
Conditioned on sets $L$ and $F$, Algorithms~\ref{alg:online_coupled} and~\ref{alg:simulated_coupled} produce the same set $N$ and assign the same weights to the elements of $N$.
\end{corollary}
\begin{proof}
Consider an element $u \in E \setminus L$. Algorithm~\ref{alg:online_coupled} adds $u$ to $N$ if and only if $u \in F$ and $\Greedy(M + u, \varnothing)$ tries to add $u$ to its solution. When this happens, $u$ is assigned the weight $f(u \mid M_u)$, where $M_u$ is the solution {\Greedy} has immediately before it process $u$. Clearly both the weight and the membership of $u$ in $N$ does not change if $\Greedy(M + u, \varnothing)$ is replaced with $\Greedy(M + u, u)$ in the above description.

On the other hand, Algorithm~\ref{alg:simulated_coupled} adds $u$ to $N$ if and only if $u \in F$ and $\Greedy(N, E \setminus L)$ tries to add $u$ to its solution. When this happens, $u$ is assigned the weight $f(u \mid M_u)$, where $M_u$ is the solution {\Greedy} has immediately before it process $u$. By Observation~\ref{obs:greedy_properties}, both $\Greedy(M + u, u)$ and $\Greedy(E, E \setminus L)$ behave the same because:
\[
	\Greedy(E, E \setminus L) = \Greedy(L, \varnothing) = M = (M + u) - u \subseteq L = E \setminus (E \setminus L) \enspace.
\]
Hence, Algorithms~\ref{alg:online_coupled} and~\ref{alg:simulated_coupled} make the same decisions regarding the membership and weight of $u$.
\end{proof}

Observation~\ref{obs:N_0_elements} together with the last corollary shows that, conditioned on $L$ and $F$, {\Linear} is fed by the same set $E \setminus L$ of elements and weights by Algorithms~\ref{alg:online_coupled} and~\ref{alg:simulated_coupled}. To complete the proof that both algorithms produce the same output distribution, conditioned on $L$ and $F$, we just need to show that the elements of $E \setminus L$ are passed to {\Linear} in the same distribution of orders under both algorithms. This is true by the following observations.
\begin{compactitem}
	\item Algorithm~\ref{alg:online_coupled} passes the elements of $E \setminus L$ in the order in which they arrive, which is uniformly random.
	\item Algorithm~\ref{alg:simulated_coupled} passes the elements of $N \cup N_0 = E \setminus L$ in a uniformly random order.
\end{compactitem}

\section{Missing Proofs}\label{app:missingProofs}

In this appendix we provide the missing proofs
for Proposition~\ref{prp:partial_problem} and 
Lemma~\ref{lem:greedy_determinisitic_property_nonnormalized}.

\subsection{Proof of Proposition~\ref{prp:partial_problem}}

\begin{repproposition}{prp:partial_problem}
Given an algorithm {\Linear} for {\SP}, there exists an algorithm for Partial-{\SP} whose competitive ratio for every matroid $\cM$ is as good as the competitive ratio of {\Linear} for this matroid.
\end{repproposition}

The algorithm we use to prove Proposition~\ref{prp:partial_problem} is Algorithm~\ref{alg:partial_algorithm}.

\begin{algorithm}[ht]
\caption{\textsf{Algorithm for Partial-{\SP}}} \label{alg:partial_algorithm}
\tcp{Initialization}
Let $L' \gets L$.\\
Let $r \gets n - |L|$.\\

\BlankLine

\tcp{Main Loop}
\While{$|L'| + r > 0$}
{
	\With{probability $\nicefrac{r}{(|L'| + r)}$}
	{
		Update $r \gets r - 1$.\\
		Pass the next arriving element to {\Linear} with its original weight $w(u)$.\\
	}
	\Otherwise
	{
		Let $u$ be a uniformly random element of $L'$.\\
		Update $L' \gets L' - u$.\\
		Pass $u$ to {\Linear} with a weight of $0$.\\
	}
}
Return the output of {\Linear} without the elements of $L$.
\end{algorithm}

\begin{observation} \label{obs:legal_input}
The input passed to {\Linear} by Algorithm~\ref{alg:partial_algorithm} is a uniformly random permutation of $E$, where the weight of an element $u \in E$ is defined by:
\[
	w'(u)
	=
	\begin{cases}
		0 & \text{if $u \in L$} \enspace,\\
		w(u) & \text{otherwise} \enspace.
	\end{cases}
\]
\end{observation}
\begin{proof}
The only part of the observation that requires a proof is the claim that the elements of $E$ are passed to {\Linear} in a uniformly random order. To see why this is true, notice that $r$ is maintained as the number of input elements that have not been passed yet to {\Linear} and $L'$ is the set of elements of $L$ that have not been passed yet to {\Linear}. Hence, at each time point the set of elements that have not been passed yet to {\Linear} is of size $r + |L'|$, and the algorithm passes a uniformly random element from this set.\end{proof}

The following lemma implies Proposition~\ref{prp:partial_problem}.

\begin{lemma}
Algorithm~\ref{alg:partial_algorithm} is an algorithm for Partial-{\SP} whose competitive ratio for every matroid $\cM$ is as good as the competitive ratio of {\Linear} for this matroid.
\end{lemma}
\begin{proof}
Assume {\Linear} is $\alpha$-competitive for an arbitrary matroid $\cM = (E, \cI)$. By Observation~\ref{obs:legal_input}, whenever Algorithm~\ref{alg:partial_algorithm} faces an instances of Partial-{\SP} over the matroid $\cM$ with a weight function $w$, it makes {\Linear} face an instance of {\SP} over the same matroid with weights given by $w'$. Hence, if we denote by $S$ the output of {\Linear}, then, by the guarantee of {\Linear}:
\[
	\bE[w'(S)]
	\geq
	\alpha \cdot \max_{T \in \cI} w'(T)
	\enspace.
\]

Recall that $w'$ is equal to $w$ for elements of $E \setminus L$ and is equal to $0$ for all other elements. Hence,
\[
	\bE[w(S \setminus L)]
	=
	\bE[w'(S)]
	\geq
	\alpha \cdot \max_{T \in \cI} w'(T)
	=
	\alpha \cdot \max_{T \in \cI \cap 2^{E \setminus L}} w'(T)
	=
	\alpha \cdot \max_{T \in \cI \cap 2^{E \setminus L}} w(T)
	\enspace.
\]
This completes the proof of the lemma since $S \setminus L$ is the output of Algorithm~\ref{alg:partial_algorithm} and $\max_{T \in \cI \cap 2^{E \setminus L}} w(T)$ is the value of the optimal solution for the instance of Partial-{\SP} faced by Algorithm~\ref{alg:partial_algorithm}.
\end{proof}

\subsection{Proof of Lemma~\ref{lem:greedy_determinisitic_property_nonnormalized}}

\begin{replemma}{lem:greedy_determinisitic_property_nonnormalized}
For a matroid and a non-negative submodular function $f$, if $S$ is the independent set returned by {\Greedy}, then for any independent set $C$, $f(S) \geq f(C \cup S) / 2$.
\end{replemma}

First, we need some definitions. Let $k = |S|$, and for every $0 \leq i \leq k$ let $S_i$ be the solution of {\Greedy} after $i$ elements are added to it. Additionally, for every $1 \leq i \leq k$, let $u_i$ be the single element in $S_i \setminus S_{i - 1}$. We also need to define the sets $\{C_i\}_{i = 0}^k$ recursively as follows:
\begin{compactitem}
	\item The set $C_0$ is simply $C$. Observe that $C_0 \cup S_0 = C$ is independent.
	\item For every $1 \leq i \leq k$, the set $C_i$ is a maximal independent subset of $C_{i - 1} + u_i$ that contains $S_i$. It is also useful to define $C'_i = C_{i - 1} \setminus C_i$, the set of elements that appear in $C_{i - 1}$ but not in $C_i$. By matroid properties, since $S_i - u_i = S_{i - 1} \subseteq C_{i - 1}$ and $C_{i - 1}$ is independent, the size of $C'_i$ is at most $1$.
\end{compactitem}

Let us now make a few observations regarding the above definitions.

\begin{observation} \label{obs:S_C_k_relation}
$f(S_k) \geq f(C_k)$.
\end{observation}
\begin{proof}
Assume towards a contradiction that $f(S_k) < f(C_k)$. Then, by submodularity since $S_k \subseteq C_k$, there must exist an element $u \in C_k \setminus S_k$ obeying $f(u \mid S_k) > 0$. On the other hand, since $C_k$ is independent we must also have that $S_k + u$ is independent. However, the existence of an element with these properties contradicts the fact that $S_k$ is the output of {\Greedy}.
\end{proof}

\begin{lemma} \label{lem:gain_bound}
For every $1 \leq i \leq k$, $f(u_i \mid S_{i - 1}) \geq f(C'_i \cup C_k) - f(C_k)$.
\end{lemma}
\begin{proof}
If $C'_i = \varnothing$, then the right hand side of the inequality we want to prove is $0$ while its left hand side is non-negative (since {\Greedy} chose to add $u_i$ to $S_{i - 1}$). Thus, we concentrate from this point on the case $|C'_i| = 1$.

Let $u'_i$ be the single element of $C'_i$. Since $C_i$ contains the set $S_{i - 1}$ and is independent, we know that $u'_i$ is not in $S_{i - 1}$ and could be added to $S_{i - 1}$ by {\Greedy} without violating independence. On the other hand, since {\Greedy} chose to add $u_i$ to $S_{i - 1}$, we must have:
\[
	f(u_i \mid S_{i - 1})
	\geq
	f(u'_i \mid S_{i - 1})
	\geq
	f(u'_i \mid C_k)
	=
	f(C'_i \cup C_k) - f(C_k)
	\enspace,
\]
where the second inequality follows by submodularity since $S_{i - 1} \subseteq S_k \subseteq C_k$.
\end{proof}

We are now ready to prove Lemma~\ref{lem:greedy_determinisitic_property_nonnormalized}.
\begin{align*}
	f(S)
	={} &
	\frac{f(S_k) + f(\varnothing) + \sum_{i = 1}^k f(u_i \mid S_{i - 1})}{2}
	\geq
	\frac{f(C_k) + f(\varnothing) + \sum_{i = 1}^k [f(C'_i \cup C_k) - f(C_k)]}{2}\\
	\geq{} &
	\frac{f(C_k) + f(\varnothing) + [f(\bigcup_{i = 1}^k C'_i \cup C_k) - f(C_k)]}{2}
	=
	\frac{f(\varnothing) + f(S \cup C)}{2}
	\geq
	\frac{f(S \cup C)}{2}
	\enspace,
\end{align*}
where the first inequality holds by Observation~\ref{obs:S_C_k_relation} and Lemma~\ref{lem:gain_bound}, the second inequality holds by submodularity of $f$ and the fact that the $C_i'$ are disjoint, and the last by the non-negativity of $f$.

}

\end{document}